\documentclass{article}

    \usepackage{graphicx}
    
    \usepackage{caption}
    \DeclareCaptionFormat{nocaption}{}
    \usepackage{float}
    \usepackage{xcolor} 
    \usepackage{enumerate} 
    \usepackage{amsmath} 
    \usepackage{amssymb} 
    \usepackage{textcomp} 
    \AtBeginDocument{%
    }
    \usepackage{upquote} 
    \usepackage{eurosym} 

    \usepackage{iftex}
    \ifPDFTeX
        \usepackage[T1]{fontenc}
        \IfFileExists{alphabeta.sty}{
              \usepackage{alphabeta}
          }{
              \usepackage[mathletters]{ucs}
              \usepackage[utf8x]{inputenc}
          }
    \else
        \usepackage{fontspec}
        \usepackage{unicode-math}
    \fi

    \usepackage{fancyvrb} 
    \usepackage{grffile} 
    \makeatletter 
    \@ifpackagelater{grffile}{2019/11/01}
    {
    }
    {
      \def\Gread@@xetex#1{%
        \IfFileExists{"\Gin@base".bb}%
        {\Gread@eps{\Gin@base.bb}}%
        {\Gread@@xetex@aux#1}%
      }
    }
    \makeatother
    \usepackage[Export]{adjustbox} 
    \adjustboxset{max size={0.9\linewidth}{0.9\paperheight}}

    \usepackage{hyperref}
    \usepackage{titling}
    \usepackage{longtable} 
    \usepackage{booktabs}  
    \usepackage{array}     
    \usepackage{calc}      
    \usepackage[inline]{enumitem} 
    \usepackage[normalem]{ulem} 
    \usepackage{mathrsfs}
    
      \usepackage{arxiv}

    \definecolor{urlcolor}{rgb}{0,.145,.698}
    \definecolor{linkcolor}{rgb}{.71,0.21,0.01}
    \definecolor{citecolor}{rgb}{.12,.54,.11}

    \definecolor{ansi-black}{HTML}{3E424D}
    \definecolor{ansi-black-intense}{HTML}{282C36}
    \definecolor{ansi-red}{HTML}{E75C58}
    \definecolor{ansi-red-intense}{HTML}{B22B31}
    \definecolor{ansi-green}{HTML}{00A250}
    \definecolor{ansi-green-intense}{HTML}{007427}
    \definecolor{ansi-yellow}{HTML}{DDB62B}
    \definecolor{ansi-yellow-intense}{HTML}{B27D12}
    \definecolor{ansi-blue}{HTML}{208FFB}
    \definecolor{ansi-blue-intense}{HTML}{0065CA}
    \definecolor{ansi-magenta}{HTML}{D160C4}
    \definecolor{ansi-magenta-intense}{HTML}{A03196}
    \definecolor{ansi-cyan}{HTML}{60C6C8}
    \definecolor{ansi-cyan-intense}{HTML}{258F8F}
    \definecolor{ansi-white}{HTML}{C5C1B4}
    \definecolor{ansi-white-intense}{HTML}{A1A6B2}
    \definecolor{ansi-default-inverse-fg}{HTML}{FFFFFF}
    \definecolor{ansi-default-inverse-bg}{HTML}{000000}

    \definecolor{outerrorbackground}{HTML}{FFDFDF}

    \providecommand{\tightlist}{%
      \setlength{\itemsep}{0pt}\setlength{\parskip}{0pt}}
    \DefineVerbatimEnvironment{Highlighting}{Verbatim}{commandchars=\\\{\}}

    \let\Oldtex\TeX
    \let\Oldlatex\LaTeX
    \renewcommand{\TeX}{\textrm{\Oldtex}}
    \renewcommand{\LaTeX}{\textrm{\Oldlatex}}

\makeatletter
\def\PY@reset{\let\PY@it=\relax \let\PY@bf=\relax%
    \let\PY@ul=\relax \let\PY@tc=\relax%
    \let\PY@bc=\relax \let\PY@ff=\relax}
\def\PY@tok#1{\csname PY@tok@#1\endcsname}
\def\PY@toks#1+{\ifx\relax#1\empty\else%
    \PY@tok{#1}\expandafter\PY@toks\fi}
\def\PY@do#1{\PY@bc{\PY@tc{\PY@ul{%
    \PY@it{\PY@bf{\PY@ff{#1}}}}}}}
\def\PY#1#2{\PY@reset\PY@toks#1+\relax+\PY@do{#2}}

\@namedef{PY@tok@w}{\def\PY@tc##1{\textcolor[rgb]{0.73,0.73,0.73}{##1}}}
\@namedef{PY@tok@c}{\let\PY@it=\textit\def\PY@tc##1{\textcolor[rgb]{0.24,0.48,0.48}{##1}}}
\@namedef{PY@tok@cp}{\def\PY@tc##1{\textcolor[rgb]{0.61,0.40,0.00}{##1}}}
\@namedef{PY@tok@k}{\let\PY@bf=\textbf\def\PY@tc##1{\textcolor[rgb]{0.00,0.50,0.00}{##1}}}
\@namedef{PY@tok@kp}{\def\PY@tc##1{\textcolor[rgb]{0.00,0.50,0.00}{##1}}}
\@namedef{PY@tok@kt}{\def\PY@tc##1{\textcolor[rgb]{0.69,0.00,0.25}{##1}}}
\@namedef{PY@tok@o}{\def\PY@tc##1{\textcolor[rgb]{0.40,0.40,0.40}{##1}}}
\@namedef{PY@tok@ow}{\let\PY@bf=\textbf\def\PY@tc##1{\textcolor[rgb]{0.67,0.13,1.00}{##1}}}
\@namedef{PY@tok@nb}{\def\PY@tc##1{\textcolor[rgb]{0.00,0.50,0.00}{##1}}}
\@namedef{PY@tok@nf}{\def\PY@tc##1{\textcolor[rgb]{0.00,0.00,1.00}{##1}}}
\@namedef{PY@tok@nc}{\let\PY@bf=\textbf\def\PY@tc##1{\textcolor[rgb]{0.00,0.00,1.00}{##1}}}
\@namedef{PY@tok@nn}{\let\PY@bf=\textbf\def\PY@tc##1{\textcolor[rgb]{0.00,0.00,1.00}{##1}}}
\@namedef{PY@tok@ne}{\let\PY@bf=\textbf\def\PY@tc##1{\textcolor[rgb]{0.80,0.25,0.22}{##1}}}
\@namedef{PY@tok@nv}{\def\PY@tc##1{\textcolor[rgb]{0.10,0.09,0.49}{##1}}}
\@namedef{PY@tok@no}{\def\PY@tc##1{\textcolor[rgb]{0.53,0.00,0.00}{##1}}}
\@namedef{PY@tok@nl}{\def\PY@tc##1{\textcolor[rgb]{0.46,0.46,0.00}{##1}}}
\@namedef{PY@tok@ni}{\let\PY@bf=\textbf\def\PY@tc##1{\textcolor[rgb]{0.44,0.44,0.44}{##1}}}
\@namedef{PY@tok@na}{\def\PY@tc##1{\textcolor[rgb]{0.41,0.47,0.13}{##1}}}
\@namedef{PY@tok@nt}{\let\PY@bf=\textbf\def\PY@tc##1{\textcolor[rgb]{0.00,0.50,0.00}{##1}}}
\@namedef{PY@tok@nd}{\def\PY@tc##1{\textcolor[rgb]{0.67,0.13,1.00}{##1}}}
\@namedef{PY@tok@s}{\def\PY@tc##1{\textcolor[rgb]{0.73,0.13,0.13}{##1}}}
\@namedef{PY@tok@sd}{\let\PY@it=\textit\def\PY@tc##1{\textcolor[rgb]{0.73,0.13,0.13}{##1}}}
\@namedef{PY@tok@si}{\let\PY@bf=\textbf\def\PY@tc##1{\textcolor[rgb]{0.64,0.35,0.47}{##1}}}
\@namedef{PY@tok@se}{\let\PY@bf=\textbf\def\PY@tc##1{\textcolor[rgb]{0.67,0.36,0.12}{##1}}}
\@namedef{PY@tok@sr}{\def\PY@tc##1{\textcolor[rgb]{0.64,0.35,0.47}{##1}}}
\@namedef{PY@tok@ss}{\def\PY@tc##1{\textcolor[rgb]{0.10,0.09,0.49}{##1}}}
\@namedef{PY@tok@sx}{\def\PY@tc##1{\textcolor[rgb]{0.00,0.50,0.00}{##1}}}
\@namedef{PY@tok@m}{\def\PY@tc##1{\textcolor[rgb]{0.40,0.40,0.40}{##1}}}
\@namedef{PY@tok@gh}{\let\PY@bf=\textbf\def\PY@tc##1{\textcolor[rgb]{0.00,0.00,0.50}{##1}}}
\@namedef{PY@tok@gu}{\let\PY@bf=\textbf\def\PY@tc##1{\textcolor[rgb]{0.50,0.00,0.50}{##1}}}
\@namedef{PY@tok@gd}{\def\PY@tc##1{\textcolor[rgb]{0.63,0.00,0.00}{##1}}}
\@namedef{PY@tok@gi}{\def\PY@tc##1{\textcolor[rgb]{0.00,0.52,0.00}{##1}}}
\@namedef{PY@tok@gr}{\def\PY@tc##1{\textcolor[rgb]{0.89,0.00,0.00}{##1}}}
\@namedef{PY@tok@ge}{\let\PY@it=\textit}
\@namedef{PY@tok@gs}{\let\PY@bf=\textbf}
\@namedef{PY@tok@gp}{\let\PY@bf=\textbf\def\PY@tc##1{\textcolor[rgb]{0.00,0.00,0.50}{##1}}}
\@namedef{PY@tok@go}{\def\PY@tc##1{\textcolor[rgb]{0.44,0.44,0.44}{##1}}}
\@namedef{PY@tok@gt}{\def\PY@tc##1{\textcolor[rgb]{0.00,0.27,0.87}{##1}}}
\@namedef{PY@tok@err}{\def\PY@bc##1{{\setlength{\fboxsep}{\string -\fboxrule}\fcolorbox[rgb]{1.00,0.00,0.00}{1,1,1}{\strut ##1}}}}
\@namedef{PY@tok@kc}{\let\PY@bf=\textbf\def\PY@tc##1{\textcolor[rgb]{0.00,0.50,0.00}{##1}}}
\@namedef{PY@tok@kd}{\let\PY@bf=\textbf\def\PY@tc##1{\textcolor[rgb]{0.00,0.50,0.00}{##1}}}
\@namedef{PY@tok@kn}{\let\PY@bf=\textbf\def\PY@tc##1{\textcolor[rgb]{0.00,0.50,0.00}{##1}}}
\@namedef{PY@tok@kr}{\let\PY@bf=\textbf\def\PY@tc##1{\textcolor[rgb]{0.00,0.50,0.00}{##1}}}
\@namedef{PY@tok@bp}{\def\PY@tc##1{\textcolor[rgb]{0.00,0.50,0.00}{##1}}}
\@namedef{PY@tok@fm}{\def\PY@tc##1{\textcolor[rgb]{0.00,0.00,1.00}{##1}}}
\@namedef{PY@tok@vc}{\def\PY@tc##1{\textcolor[rgb]{0.10,0.09,0.49}{##1}}}
\@namedef{PY@tok@vg}{\def\PY@tc##1{\textcolor[rgb]{0.10,0.09,0.49}{##1}}}
\@namedef{PY@tok@vi}{\def\PY@tc##1{\textcolor[rgb]{0.10,0.09,0.49}{##1}}}
\@namedef{PY@tok@vm}{\def\PY@tc##1{\textcolor[rgb]{0.10,0.09,0.49}{##1}}}
\@namedef{PY@tok@sa}{\def\PY@tc##1{\textcolor[rgb]{0.73,0.13,0.13}{##1}}}
\@namedef{PY@tok@sb}{\def\PY@tc##1{\textcolor[rgb]{0.73,0.13,0.13}{##1}}}
\@namedef{PY@tok@sc}{\def\PY@tc##1{\textcolor[rgb]{0.73,0.13,0.13}{##1}}}
\@namedef{PY@tok@dl}{\def\PY@tc##1{\textcolor[rgb]{0.73,0.13,0.13}{##1}}}
\@namedef{PY@tok@s2}{\def\PY@tc##1{\textcolor[rgb]{0.73,0.13,0.13}{##1}}}
\@namedef{PY@tok@sh}{\def\PY@tc##1{\textcolor[rgb]{0.73,0.13,0.13}{##1}}}
\@namedef{PY@tok@s1}{\def\PY@tc##1{\textcolor[rgb]{0.73,0.13,0.13}{##1}}}
\@namedef{PY@tok@mb}{\def\PY@tc##1{\textcolor[rgb]{0.40,0.40,0.40}{##1}}}
\@namedef{PY@tok@mf}{\def\PY@tc##1{\textcolor[rgb]{0.40,0.40,0.40}{##1}}}
\@namedef{PY@tok@mh}{\def\PY@tc##1{\textcolor[rgb]{0.40,0.40,0.40}{##1}}}
\@namedef{PY@tok@mi}{\def\PY@tc##1{\textcolor[rgb]{0.40,0.40,0.40}{##1}}}
\@namedef{PY@tok@il}{\def\PY@tc##1{\textcolor[rgb]{0.40,0.40,0.40}{##1}}}
\@namedef{PY@tok@mo}{\def\PY@tc##1{\textcolor[rgb]{0.40,0.40,0.40}{##1}}}
\@namedef{PY@tok@ch}{\let\PY@it=\textit\def\PY@tc##1{\textcolor[rgb]{0.24,0.48,0.48}{##1}}}
\@namedef{PY@tok@cm}{\let\PY@it=\textit\def\PY@tc##1{\textcolor[rgb]{0.24,0.48,0.48}{##1}}}
\@namedef{PY@tok@cpf}{\let\PY@it=\textit\def\PY@tc##1{\textcolor[rgb]{0.24,0.48,0.48}{##1}}}
\@namedef{PY@tok@c1}{\let\PY@it=\textit\def\PY@tc##1{\textcolor[rgb]{0.24,0.48,0.48}{##1}}}
\@namedef{PY@tok@cs}{\let\PY@it=\textit\def\PY@tc##1{\textcolor[rgb]{0.24,0.48,0.48}{##1}}}

\makeatother

    \definecolor{incolor}{rgb}{0.0, 0.0, 0.5}
    \definecolor{outcolor}{rgb}{0.545, 0.0, 0.0}

    \hypersetup{
      breaklinks=true,  
      colorlinks=true,
      urlcolor=urlcolor,
      linkcolor=linkcolor,
      citecolor=citecolor,
      }
    
\begin{document}
    
    \maketitle

	\begin{abstract}

\noindent Fukasawa introduced in [Fukasawa, Math Financ, 2012] two necessary conditions for no butterfly arbitrage which require that the $d_1$ and $d_2$ functions of the Black-Scholes formula have to be decreasing. In this article we characterize the set of smiles satisfying these conditions, using the parametrization of the smile in delta. We obtain a parametrization of the set via one real number and three positive functions. We also show that such smiles and their symmetric smiles can be transformed into smiles in the strike space by a bijection. Our result motivates the study of the challenging question of characterizing the subset of butterfly arbitrage-free smiles.

\end{abstract}

	\hypertarget{introduction}{%
\section{Introduction}\label{introduction}}

	FX OTC options are quoted in delta through the At-The-Money volatility,
and Risk Reversal and Strangle prices for different delta points. On
Equity and Commodity markets, many trading firms analyze the risk of
options portfolios at a given maturity on a grid of deltas instead of a
grid of strikes. This fundamentally relates to the high-view of the risk
as being driven at first order by a delta risk and a Vega risk, which
underpins in particular CME's SPAN methodology, which, although dated
back to 1980, is still widely used in CCPs.

	From this perspective, a first immediate question is whether the
transformation of a smile in the strike space to a smile in the delta
space is well-defined. We show that it is the case under the assumption
that the map \(k \to \delta(k):=N(d_1(k,\hat\sigma(k))\) is decreasing,
where \(d_1\) denotes the classical quantity in the Black-Scholes
formula, \(N\) the Gaussian cumulative density function, and \(k\) the
log-forward moneyness. Thanks to a result of Fukasawa
\cite{fukasawa2012normalizing}, we know that this property is fulfilled
by smiles with no butterfly arbitrage. Given the fact that a smile has
no butterfly arbitrage if and only if the symmetrical smile (in log
forward-moneyness) has no butterfly arbitrage, such smiles will also
have the property that the map
\(k \to \delta(k):=N(d_2(k,\hat\sigma(k))\) is decreasing, since the
\(d_1\) function for the symmetrical smile is minus the \(d_2\) function
for the original smile evaluated at the strike point corresponding to
minus the original log forward-moneyness point.

The second immediate question is how the no butterfly arbitrage
condition translates in the delta space. Surprisingly, this question is
essentially an open one. Some results on the absence of arbitrage for
implied volatility surfaces in the delta space can be found in \cite{lucic2021normalizing}.
We don't address the delta no butterfly arbitrage in this work, but we consider instead a weaker
condition which is that the 2 mappings
\(k \to d_{1,2}(k,\hat\sigma(k))\) are decreasing, and obtain an
explicit parameterization of the smiles in delta fulfilling those
conditions. This family can be useful in practice, since such smiles are
expected to be not too far from fully (strongly) no arbitrage ones.

	From a bird's eye theoretical view, another take on our work is to
consider the open question of parameterizing all the smiles with no
butterfly arbitrage. Such a parameterization would allow practitioners
to calibrate implied volatility smiles with the guaranty of fulfilling
no butterfly arbitrage and without range restrictions coming from the
usage of a particular model (like Gatheral's SVI, see for example
\cite{gatheral2011volatility}). This question is probably a difficult
one, and our work can be seen as a solution to the same question for a
notion of \emph{weak arbitrage}, through a move to the delta space. This
could suggest that there is some hope to solve also the initial question
in the delta space.

	We start in \cref{notations-and-preliminaries} with a description of the
delta notation and a detailed discussion of the no butterfly arbitrage
conditions in the delta space.

In \cref{from-a-smile-in-delta-to-a-smile-in-k} we characterize the set
of smiles in delta that can be converted to smiles in log-forward
moneyness, i.e.~the set of smiles that allow to unambiguously define a
delta function \(k\to\delta(k)\) from the relation
\(\delta(k)=N(d_1(k,\sigma(\delta(k)))\). Looking at a similar question
but from the strike perspective,
\cref{from-a-smile-in-k-to-a-smile-in-delta} achieves the
characterization of the set of smiles in delta which correspond to an
existing smile in strike, i.e.~smiles that are defined from a strike
function \(\delta\to k(\delta)\) satisfying
\(\delta=N(d_1(k(\delta),\hat\sigma(k(\delta)))\). The two sets are
shown to coincide, so that a smile in delta belonging to them can be
transformed into a smile in strike and re-transformed in the original
smile in delta.

Since arbitrage-free smiles have symmetrical smiles (in log-forward
moneyness) which are still arbitrage-free,
\cref{symmetric-transformation} deals with the set of smiles in delta
which satisfy the weak no arbitrage conditions of monotonicity of the
functions \(d_1\) and \(d_2\). In particular, the result in
\cref{theoParamT1} shows that such set can be parametrized by a real
number and three positive functions. The section ends with practical
examples of smiles in the weak no arbitrage set.

	\hypertarget{notations-and-preliminaries}{%
\section{Notations and
preliminaries}\label{notations-and-preliminaries}}

	In \cref{TableNotations} we summarize the notations that will be used in
the article.

\begin{table*}[!b]
    \begin{center}
        \begin{tabularx}{\linewidth}{ c X }
            \hline
            Symbol & Meaning\\
            \hline
            $\sigma$ & Smile in delta\\
            $\hat\sigma$ & Smile in strike\\
            $\Sigma$ & Set of smiles in delta\\
            $\mathcal D(\mathcal A,\mathcal B)$ & Set of continuous functions from $\mathcal A$ to $\mathcal B$ which admit a left/right derivative everywhere\\
            $\delta\to k$ & Subscript for sets and functions related to the transformation of a smile in delta into a smile in strike\\
            $k\to\delta$ & Subscript for sets and functions related to the transformation of a smile in strike into a smile in delta\\
            \hline
        \end{tabularx}
        \caption{\label{TableNotations} Symbols used in the article and their relative meaning.}
    \end{center}
\end{table*}

	For a fixed maturity \(T\), we denote by \(C_{BS}(k,\hat\sigma(k))\) the
Black-Scholes pricing formula for a call option with maturity \(T\),
strike \(F_0(T)e^{k}\), forward value \(F_0(T)\), discounting factor
\(D_0(T)\), and implied volatility \(\hat\sigma(k)\): \begin{align*}
C_{BS}(k,\hat\sigma(k)) &= D_0(T)F_0(T)\bigl(N(d_1(k,\hat\sigma(k)) - e^kN(d_2(k,\hat\sigma(k))\bigr),\\
d_{1,2}(k,\hat\sigma(k)) &= -\frac{k}{\hat\sigma(k)\sqrt{T}} \pm \frac{\hat\sigma(k)\sqrt{T}}{2}.
\end{align*} For easier notation, we will sometimes denote with
\(d_1(k)\) and \(d_2(k)\) the functions \(d_1(k,\hat\sigma(k))\) and
\(d_2(k,\hat\sigma(k))\) respectively.

	Let a number \(C(k)\) lie strictly between \(D_0(T)(F_0(T)-K)^+\) and
\(D_0(T)F_0(T)\). Then the \emph{implied volatility} \(\hat \sigma(k)\)
is well defined by the property \(C_{BS}(k, \hat \sigma(k)) = C(k)\). In
turn, this defines the quantity \emph{delta} by
\begin{equation}\label{eqDefDelta}
\delta(k):= N(d_1(k, \hat \sigma(k))).
\end{equation}

	\begin{remark}\label{remarkFromDeltaToK}

Observe that the pair $(\delta(k), \hat\sigma(k))$ allows to recover the log-forward moneyness $k$, indeed
$$k = \Bigl(-N^{-1}(\delta(k))+\frac{\hat\sigma(k) \sqrt T}2\Bigr)\hat\sigma(k) \sqrt T.$$

\end{remark}

	\hypertarget{no-butterfly-arbitrage}{%
\subsection{No butterfly arbitrage}\label{no-butterfly-arbitrage}}

	Under the hypothesis of a perfect market for the underlying asset and
for the Call options, using e.g.~\cite{tehranchi2020black}, a Call price
function with respect to the strike is free of butterfly arbitrage iff
it is

\begin{enumerate}
\def\labelenumi{\arabic{enumi}.}
\tightlist
\item
  convex,
\item
  non-increasing,
\item
  contained in the interval \(]D_0(T)(F_0(T)-K)^+, D_0(T)F_0(T)[\).
\end{enumerate}

In the case of Call prices specified with the Black-Scholes formula
through an implied volatility, the third property is automatically
satisfied. Indeed, the Black-Scholes formula is increasing with respect
to the implied volatility and it tends to the two bounds when the
implied volatility goes to \(0\) and \(\infty\) respectively. Given that
the third property is granted, then the first property implies the
second since an increasing and convex function cannot be bounded.

Therefore, in our context, there is no butterfly arbitrage iff Call
prices are convex.

As shown in Theorem 2.9 condition (IV3) of \cite{roper2010arbitrage}, in
the case of twice derivable implied volatility functions, the
requirement of prices' convexity corresponds to the requirement that the
function \begin{equation}\label{eqButtArbitrage}
\hat\sigma''(k) + d_1'(k,\hat\sigma(k))d_2'(k,\hat\sigma(k))\hat\sigma(k)
\end{equation} is non-negative. Such requirement is called the
Durrleman's Condition.

	\hypertarget{assumptions}{%
\subsection{Assumptions}\label{assumptions}}

	As shown in section 2.3 of \cite{martini2022no}, the condition of
vanishing Call prices for increasing strikes is not necessary. Such
condition is one-to-one with the behavior of the function \(d_1\) at
\(\infty\). In particular, it holds true iff
\begin{equation}\label{eqLimitD1}
\lim_{k\to\infty}d_1(k,\hat\sigma(k)) = -\infty.
\end{equation} This follows from the fact that the value of a Call price
at \(\infty\) is the normal quantile in \(d_1(\infty)\):
\begin{equation}\label{eqCInfty}
C(\infty) = D_0(T)F_0(T)N^{-1}(d_1(\infty))
\end{equation} as shown in point (A3) in the proof of Theorem 2.9 of
\cite{roper2010arbitrage}.

Furthermore, the bounds of the Call prices function imply that its slope
for \(K\) going to \(0\) lies between \(-1\) and \(0\). The limit
corresponds to \(-1\) iff there is no mass of the underlying in \(0\).
This in turn is equivalent to the fact that the left limit of the
function \(d_2\) is \(\infty\): \begin{equation}\label{eqLimitD2}
\lim_{k\to-\infty}d_2(k,\hat\sigma(k)) = \infty.
\end{equation} This follows from the fact that
\(d_2(-\infty)=-N^{-1}(P(S_T=0))\) as shown in proposition 2.4 of
\cite{fukasawa2010normalization}. Now, it holds
\(C(K)=\int_K^{\infty}(x-K)p_{S_T}(x)\,dx\) where \(p_{S_T}\) is the
probability distribution of the underlying price at maturity. By the
Leibniz integral rule, it holds \begin{align*}
C'(K) &= \int_K^{\infty}\frac{d}{dK}\bigl((x-K)p_{S_T}(x)\bigr)\,dx - (x-K)p_{S_T}(x)\bigl|_{x=K}\\
&= -\int_K^{\infty}p_{S_T}(x)\,dx\\
&= -P(S_T>K).
\end{align*} Letting \(K\) go to \(0\), this implies
\(P(S_T>0)=-C'(0)\), so that \(P(S_T=0)=1+C'(0)\). Then
\begin{equation}\label{eqCDer0}
d_2(-\infty)=-N^{-1}(1+C'(0)).
\end{equation}

	\begin{remark}\label{remarkHypothesis}

In the present article, we will take \cref{eqLimitD1,eqLimitD2} as assumptions, even though they are not necessary conditions for butterfly arbitrage-free Call prices. The first assumption implies that Call prices vanish for increasing strikes while the second assumption implies that Call prices have a slope of $-1$ for null strikes and that there is no mass of the underlying in $0$.

\end{remark}

The rationale of these assumptions relies on the fact that we will
require the delta function to have as image the whole interval
\(]0,1[\), as we will show in the below section. This requirement avoids
degenerate cases that may be less appealing to practitioners.

	Assumptions in \cref{remarkHypothesis} imply some easy consequences that
we state in the following proposition. The Lee conditions can be found
in \cite{lee2004moment}, Lemmas 3.1 and 3.3.

\begin{proposition}\label{PropLee}

Assumptions in \cref{remarkHypothesis} hold iff $d_1(k)$ and $d_2(k)$ are surjective. In such case, it holds:
\begin{itemize}
    \item Lee right wing condition: $\hat\sigma(k)\sqrt{T}<\sqrt{2k}$ for $k\gg0$;
    \item Lee left wing condition: $\hat\sigma(k)\sqrt{T}<\sqrt{-2k}$ for $k\ll0$.
\end{itemize}

\end{proposition}

\begin{proof}

For simple algebraic reasons (the arithmetic mean exceeds the geometric mean) the function $k \to d_1(k, \hat\sigma(k))$ is greater than $\sqrt{-2k}$ for every $k\leq0$ (see Lemma 3.5 of \cite{fukasawa2012normalizing}). As a consequence, $d_1(k, \hat\sigma(k))$ always goes to $\infty$ as $k$ goes to $-\infty$. Similarly with the same proof, the function $d_2(k,\hat\sigma(k))$ always goes to $-\infty$ at $\infty$ since it is smaller than $-\sqrt{2k}$. Then, under hypothesis in \cref{remarkHypothesis}, the functions $d_1(k)$ and $d_2(k)$ are surjective. The if implication is trivial.

Since the function $k\to d_1(k,\hat\sigma(k))$ goes to $-\infty$ on the right under \cref{remarkHypothesis}, it must be negative for $k$ large enough, which implies the right wing Lee condition. Similarly, since the function $k\to d_2(k,\hat\sigma(k))$ explodes on the left, it must be positive for $k$ small enough and the left wing Lee condition holds.

\end{proof}

Observe that Lee shows that the left wing condition holds for every
arbitrage-free smile iff \(P(S_T=0)<\frac{1}2\), which is indeed our
case since we suppose no mass in \(0\). Furthermore, the proof given for
only if implication can also be found in Lemmas 3.2 and 3.5 of
\cite{fukasawa2012normalizing}).

	In the present article we will often use necessary conditions for the
absence of butterfly arbitrage found by Fukasawa in Theorem 2.8 of
\cite{fukasawa2012normalizing}. We recall them in the following remark.

\begin{remark}[Weak no butterfly arbitrage conditions/Fukasawa necessary no butterfly arbitrage conditions]

If $k \to C(k)$ has no butterfly arbitrage, then $k \to d_1(k, \hat\sigma(k))$ and $k \to d_2(k, \hat\sigma(k))$ are strictly decreasing.

\end{remark}

	\hypertarget{no-butterfly-arbitrage-in-the-delta-notation}{%
\subsection{No butterfly arbitrage in the delta
notation}\label{no-butterfly-arbitrage-in-the-delta-notation}}

	As discussed above, a difficult challenge is to parameterize the set of
implied volatility functions corresponding to functions \(k \to C(k)\)
with no butterfly Arbitrage. We will denote with \(\Sigma_\text{A}\) the
set of such implied volatility functions (in the delta notation
\(\delta \to \sigma(\delta)\)) satisfying also hypothesis in
\cref{remarkHypothesis}. In the following, we will show that these delta
smiles can be transformed into moneyness smiles and, given the above
study, they will guarantee the monotonicity and surjectivity of
functions \(d_1(k,\hat\sigma(k))\) and \(d_2(k,\hat\sigma(k))\). Even
though we will not reach the aim of parameterizing the set
\(\Sigma_\text{A}\), we will be able to achieve the parameterization of
the larger set of smiles satisfying the two conditions on the functions
\(d_1(k,\hat\sigma(k))\) and \(d_2(k,\hat\sigma(k))\).

	It is easy to see that the function \(\delta(k)\) as defined in
\cref{eqDefDelta} goes to \(0\) as \(k\) goes to \(\infty\) iff
\cref{eqLimitD1} holds. The condition \(\delta(-\infty)=1\) corresponds
to the fact that \(d_1(k, \hat\sigma(k))\) goes to \(\infty\) as \(k\)
decreases, but this is already verified by arbitrage-free smiles as
shown above. More precisely, from \cref{eqCInfty} it follows
\[C(\infty) = D_0(T)F_0(T)\delta(\infty)\] so that the range of
\(\delta(k)\) is \(\bigl]\frac{C(\infty)}{D_0(T)F_0(T)},1\bigr[\).

Secondly, the assumption \cref{eqLimitD2} is the equivalent of the
assumption \cref{eqLimitD1} for the symmetric smile
\(\hat{\bar\sigma}(k)=\hat\sigma(-k)\). Indeed, it holds
\(d_1(k,\hat{\bar\sigma}(k))=-d_2(-k,\hat\sigma(-k))\), so that
condition \cref{eqLimitD2} is equivalent to requiring that the \(d_1\)
function for the symmetric smile
\(\bar d_1(k):=d_1(k,\hat{\bar\sigma}(k))\) satisfies \cref{eqLimitD1}.
Denoting with \(\bar\delta(k)\) the function \(N(\bar d_1(k))\), using
\cref{eqCDer0} it is immediate that \[\bar\delta(\infty)=1+C'(0).\] On
the other hand, the limit of \(\bar\delta(k)\) for \(k\) decreasing is
still \(1\). Indeed, the symmetric smile is arbitrage-free iff the
original smile is, so from non-arbitrage arguments, the limit of
\(\bar d_1(k)\) is \(\infty\). Then, the range of the symmetric delta
\(\bar\delta(k)\) is \(]1+C'(0),1[\).

The relations between \(\hat\sigma\) and its symmetric correspondent,
could have been reached also looking at Theorem 2.2.6. of
\cite{tehranchi2020black}. Indeed, it is shown that
\(P(S_T=0)=\frac{\bar C(\infty)}{\bar D_0(T)\bar F_0(T)}\) where bars
denote prices for the symmetric process. From the above section, this
probability is also equivalent to \(1-C'(0)\) and from this section,
\(\frac{\bar C(\infty)}{\bar D_0(T)\bar F_0(T)}\) is equivalent to
\(\bar\delta(\infty)\). Then \(\bar\delta(\infty)=1+C'(0)\).

	Under assumptions in \cref{remarkHypothesis}, ranges for \(\delta\) and
\(\bar\delta\) are both \(]0,1[\).

	\hypertarget{from-a-smile-in-delta-to-a-smile-in-k}{%
\section{\texorpdfstring{From a smile in \(\delta\) to a smile in
\(k\)}{From a smile in \textbackslash delta to a smile in k}}\label{from-a-smile-in-delta-to-a-smile-in-k}}

	In the industry, it is not an unusual practice to calibrate volatility
smiles in the delta parameterization, instead of the usual strike one,
especially when dealing with Forex products. When options are quoted on
a grid of maturities and deltas, such choice is natural and easy to be
exploited. When, on the other hand, options are quoted on a grid of
maturities and strikes, the permutation between delta and strike smiles
is not straightforward.

Indeed, it is firstly necessary to transform data in strikes into data
in deltas. For a fixed maturity, the procedure reads:

\begin{enumerate}
\def\labelenumi{\arabic{enumi}.}
\tightlist
\item
  compute volatilities \(\hat\sigma(k)\) for the quoted strikes
  \(F_0(T)e^k\) using the inversion of the Black-Scholes pricing
  formula;
\item
  compute corresponding deltas \(\delta(k)\) with \cref{eqDefDelta} and
  uniquely associate them to the smile values;
\item
  interpolate pairs
  \((\delta(k),\hat\sigma(k)) = (\delta,\sigma(\delta))\) with a chosen
  method in order to recover the continuous smile
  \(\delta\to\sigma(\delta)\).
\end{enumerate}

At this point, different operations can be done using the smiles in
delta, such as collecting historical values, stressing data, doing
statistics, and so on. When it is necessary to come back to the strike
notation, for example to compute a stressed option price with known
strike, the smile in delta must have the ability to be converted into a
smile in strike.

	Since arbitrage-free call prices are convex, they admit a left/right
derivative everywhere, and so do their implied volatilities. For this
reason, we will work with functions that admit left and right
derivatives. In the following, when talking about derivative and using
the appendix \('\), we refer to either the left or the right derivative.
For this aim, as anticipated in \cref{TableNotations}, let us define the
set \(\mathcal D\) of continuous functions from \(\mathcal A\) to
\(\mathcal B\) which admit left/right derivative everywhere:
\[\mathcal D(\mathcal A,\mathcal B) := \bigl\{f:\mathcal A\longrightarrow\mathcal B\, \bigl|\, \text{$f$ continuous and with a left/right derivative everywhere}\bigr\}.\]

In this section we aim to find conditions under which any positive and
derivable at left/right function \(\delta \to \sigma(\delta)\) defined
on \(]0,1[\) allows to recover a function \(k \to \hat \sigma(k)\). We
call \(\Sigma_{\delta\to k}\) the set of smiles
\(\delta\to\sigma(\delta)\) with this property, in particular:
\begin{align*}
\Sigma_{\delta\to k} := \bigl\{\delta\to\sigma(\delta)\in \mathcal D(]0,1[,\mathbb{R}^+)\, \bigl|\, \forall k\in\mathbb R\, \exists! \delta_{\delta\to k}(k)\, |\, \delta_{\delta\to k}(k) = N(d_1(k, \sigma(\delta_{\delta\to k}(k)))),&\\
\{\delta_{\delta\to k}(k)\, |\, k\in\mathbb R\} = ]0,1[\bigr\}.&
\end{align*}

The condition of a.s. differentiability is not necessary for implied
volatility smiles. We could suppose it under the belief that it is not a
too restrictive condition. A result regarding derivatives of the implied
volatility can be found in Theorem 5.1 of \cite{rogers2010can} where the
authors show the existence of the right and left derivatives of the
implied volatility with some arbitrage bounds.

For a delta smile in \(\Sigma_{\delta\to k}\), the corresponding smile
in \(k\) is defined as
\(\hat\sigma(k):=\sigma(\delta_{\delta\to k}(k))\). Note that there is
indeed a question, because in the above procedure it could happen that
two different pairs \((\delta, \sigma(\delta))\) produce the \emph{same}
strike, meaning that there is no way to define the value
\(\hat \sigma(k)\). The second condition defining
\(\Sigma_{\delta\to k}\) is added in order to avoid degenerate cases not
useful for practitioners.

	From the first condition of \(\Sigma_{\delta\to k}\), a function
\(\delta\to\sigma(\delta)\) in \(\Sigma_{\delta\to k}\) allows to define
a function \(k\to\delta_{\delta\to k}(k)\), and from the second
condition, such function is surjective. It is worth observing that the
function \(k\to\delta_{\delta\to k}(k)\) is also injective. Indeed, if
two values \(k_1\) and \(k_2\) have the same delta
\(\delta_{\delta\to k}(k_{1})=\delta_{\delta\to k}(k_{2})\), then they
also have the same volatility since
\(\hat\sigma(k_1)=\sigma(\delta_{\delta\to k}(k_1))=\sigma(\delta_{\delta\to k}(k_2))=\hat\sigma(k_2)\).
From \cref{remarkFromDeltaToK}, the pair
\((\delta_{\delta\to k}(k_1),\hat\sigma(k_1)) = (\delta_{\delta\to k}(k_2),\hat\sigma(k_2))\)
is associated to a unique log-forward moneyness, so that \(k_1=k_2\).
Since \(\delta\to\sigma(\delta)\) is continuous and since we have proven
that the function \(k\to\delta_{\delta\to k}(k)\) is one-to-one, it is
also monotone, which implies it is a.s. differentiable.

	Let us define the function \begin{equation}\label{eqDefL}
l(\delta):=\Bigl(N^{-1}(\delta)-\frac{\sigma(\delta) \sqrt T}2\Bigr)\sigma(\delta) \sqrt T.
\end{equation}

In the following Lemma, we characterize the set \(\Sigma_{\delta\to k}\)
through its function \(l\).

\begin{lemma}

It holds
$$\Sigma_{\delta\to k} = \bigl\{\delta\to\sigma(\delta)\in \mathcal D(]0,1[,\mathbb{R}^+)\, |\, l\ \text{strictly increasing surjective}\bigr\}.$$

\end{lemma}

\begin{proof}

Reformulating the conditions of $\Sigma_{\delta\to k}$, we ask that given a continuous positive function $\delta \to \sigma(\delta)$ defined on $]0,1[$,
\begin{align*}
&\forall k\in\mathbb R\, \exists! \delta_{\delta\to k}(k)\, |\, -k = l(\delta_{\delta\to k}(k)),\\
& \bigl\{\delta_{\delta\to k}(k)\, |\, k\in\mathbb R\bigr\} = ]0,1[.
\end{align*}
This implies that the function $k\to l(\delta_{\delta\to k}(k))$ must be well-defined, monotone and surjective. For a delta smile in $\Sigma_{\delta\to k}$, the function $k\to\delta_{\delta\to k}(k)$ is monotone and surjective, so that we are requiring $\delta\to l(\delta)$ to be monotone and surjective from the interval $]0,1[$ to $]-\infty,+\infty[$. Since $l(0)$ is negative, $l$ must be strictly increasing.

These conditions are necessary but also sufficient. Indeed, the uniqueness of $\delta_{\delta\to k}(k)$ follows from the fact that if for a fixed $k$ there are two different $\delta_{\delta\to k}(k)$, say $\delta_1$ and $\delta_2$, then $l(\delta_1)\neq l(\delta_2)$ since $l$ is monotone. However, $l(\delta_{\delta\to k}(k))=-k$, so the two values of $l$ should be the same. The existence of $\delta_{\delta\to k}(k)$ is guaranteed by the surjectivity of $l$.

For the second condition, firstly observe that from the relation $-k=l(\delta_{\delta\to k}(k))$, it follows $l(\delta_{\delta\to k}(\infty))=-\infty$ and since $l$ is injective, it must hold $\delta_{\delta\to k}(\infty)=0$. Similarly to show $\delta_{\delta\to k}(-\infty)=1$. These observations guarantee the full range of $\delta_{\delta\to k}(k)$.

\end{proof}

	\hypertarget{from-a-smile-in-k-to-a-smile-in-delta}{%
\section{\texorpdfstring{From a smile in \(k\) to a smile in
\(\delta\)}{From a smile in k to a smile in \textbackslash delta}}\label{from-a-smile-in-k-to-a-smile-in-delta}}

	Starting from a smile in log-forward moneyness \(k\to\hat\sigma(k)\) and
making \(k\) move in \(\mathbb R\), one obtains a collection of pairs
\((\delta(k), \hat\sigma(k))\). There is no guarantee that such a
collection allows to define a function \(\delta \to \sigma(\delta)\).
Indeed, in order to define a function in \(\delta\), for every
\(\delta\in]0,1[\), there must exist a unique \(k\) such that
\(\delta=N(d_1(k,\hat\sigma(k))\). In such a way one can define
unambiguously a function \(\sigma\) by the equality
\(\sigma(\delta)=\hat\sigma(k)\) for all \(k\)s. We define
\(\Sigma_{k\to\delta}\) as the set of delta smiles obtained from a
strike smile: \begin{align*}
\Sigma_{k\to\delta} := \bigl\{\delta\to\sigma(\delta)\in \mathcal D(]0,1[,\mathbb{R}^+)\, \bigl|\, \exists k\to\hat\sigma(k)\in \mathcal D(\mathbb R,\mathbb{R}^+)\ \text{s.t.}\ \forall\delta\in]0,1[\, \exists!k_{k\to\delta}(\delta)\ \text{s.t.}&\\
\delta=N(d_1(k_{k\to\delta}(\delta),\hat\sigma(k_{k\to\delta}(\delta))),&\\
\hat\sigma(k_{k\to\delta}(\delta))=\sigma(\delta),&\\
\bigl\{k_{k\to\delta}(\delta)\, |\, \delta\in]0,1[\bigr\} = \mathbb R&\bigr\}.
\end{align*}

This definition can be simplified thanks to the following Lemma:

	\begin{lemma}\label{LemmaSigmaKToDelta}

It holds
\begin{align*}
\Sigma_{k\to\delta} = \bigl\{\delta\to\sigma(\delta)\in \mathcal D(]0,1[,\mathbb{R}^+)\, \bigl|\, \exists k\to\hat\sigma(k)\in \mathcal D(\mathbb R,\mathbb{R}^+)\ \text{s.t.}&\\
k\to d_1(k,\hat\sigma(k))\ \text{strictly decreasing surjective},&\\
\hat\sigma(k_{k\to\delta}(\delta))=\sigma(\delta)&\bigr\}.
\end{align*}

\end{lemma}

	\begin{proof}

The existence and uniqueness of a $k_{k\to\delta}$ in the first condition defining $\Sigma_{k\to\delta}$ can be translated in requiring that the map $k \to N(d_1(k, \hat \sigma(k)))$ is strictly monotone. The second condition defining $\Sigma_{k\to\delta}$ requires the surjectivity of $k_{k\to\delta}$. Equivalently, the two conditions hold iff the map $k \to d_1(k, \hat \sigma(k))$ is strictly monotone and surjective in $]-\infty,\infty[$. The decreasing direction is due to the formula defining the map $k\to d_1(k,\hat\sigma(k))$.

\end{proof}

	A function \(\delta\to\sigma(\delta)\in \Sigma_{k\to\delta}\) allows to
define a function \(k\to\delta_{k\to\delta}(k)\) such that
\(\delta_{k\to\delta}(k)=N(d_1(k,\hat\sigma(k)))\), and such
\(\delta_{k\to\delta}(k)\) is bijective, with inverse
\(k_{k\to\delta}(\delta)=d_1^{-1}(N^{-1}(\delta))\). Since \(d_1(k)\) is
strictly decreasing, it is a.s. differentiable and so is its inverse, so
that the function \(k_{k\to\delta}(\delta)\) is a.s. differentiable.

	We now look at the relation between smiles in delta that can be
transformed into smiles in strike and smiles in strike that can be
transformed in smiles in delta. It turns out that the image of the
latter set in the space of smiles in delta actually coincides with the
former set. In other words, a smile in delta obtained through a smile in
strike can be re-transformed into the original smile in strike. Also,
any smile in delta that can be transformed in a smile in strike can be
recovered from its transformation into a smile in strike.

	\begin{proposition}\label{PropEqualSets}

It holds $\Sigma_{\delta\to k}=\Sigma_{k\to\delta}$ and $\delta_{\delta\to k}=\delta_{k\to\delta}$.

\end{proposition}

	\begin{proof}

We firstly prove $\Sigma_{k\to\delta}\subset \Sigma_{\delta\to k}$. For a function $\delta\to\sigma(\delta)$ in $\Sigma_{k\to\delta}$ and a fixed $k\in\mathbb{R}$, there exists a unique $\delta_{k\to\delta}(k)$ such that $\delta_{k\to\delta}(k) = N(d_1(k,\hat\sigma(k))$. Given the definition of $\sigma(\delta)$, it holds $\hat\sigma(k)=\sigma(\delta_{k\to\delta}(k))$, so that $\delta_{k\to\delta}(k) = N(d_1(k,\sigma(\delta_{k\to\delta}(k)))$. Suppose there is a second $\tilde\delta$ such that $\tilde\delta = N(d_1(k,\sigma(\tilde\delta)))$. For such $\tilde\delta$ there is a unique $k_{k\to\delta}(\tilde\delta)$ such that $\tilde\delta = N(d_1(k_{k\to\delta}(\tilde\delta),\hat\sigma(k_{k\to\delta}(\tilde\delta))))$, furthermore $\hat\sigma(k_{k\to\delta}(\tilde\delta))=\sigma(\tilde\delta)$. Then, $N(d_1(k,\sigma(\tilde\delta))) = \tilde\delta = N(d_1(k_{k\to\delta}(\tilde\delta),\sigma(\tilde\delta)))$, from which it immediately follows $k=k_{k\to\delta}(\tilde\delta)$. Then for $k$, it both holds $\delta_{k\to\delta}(k) = N(d_1(k,\hat\sigma(k))$ and $\tilde\delta = N(d_1(k,\hat\sigma(k)))$, so that $\delta_{k\to\delta}(k)=\tilde\delta$. The function $\delta_{k\to\delta}(k)$, in particular, corresponds to the function $\delta_{\delta\to k}(k)$ defining the set $\Sigma_{\delta\to k}$.

On the other hand, we now look at the relation $\Sigma_{\delta\to k}\subset \Sigma_{k\to\delta}$. Let $\delta\to\sigma(\delta)$ in $\Sigma_{\delta\to k}$, then we can define a function $k\to\hat\sigma(k)$ such that $\hat\sigma(k):=\sigma(\delta_{\delta\to k}(k))$, where $\delta_{\delta\to k}(k)$ is the only $\delta$ satisfying $\delta=N(d_1(k,\sigma(\delta))$. The function $d_1(k)=d_1(k,\hat\sigma(k))$ is surjective since for any $\delta\in]0,1[$ there exists a $k$ such that $\delta=\delta_{\delta\to k}(k)$ so that $\delta=N(d_1(k,\sigma(\delta))=N(d_1(k,\hat\sigma(k))$. To prove the monotonicity, observe that for any $k$ there is a unique $\delta_{\delta\to k}(k)$ such that $\delta_{\delta\to k}(k) = N(d_1(k,\sigma(\delta_{\delta\to k}(k))))$, which is equivalent to write
\begin{equation}\label{eqllog}
l(\delta_{\delta\to k}(k))=-k.
\end{equation}
Also, from the definition of $\hat\sigma(k)$, it holds
\begin{equation}\label{eqCharS1}
\delta_{\delta\to k}(k) = N(d_1(k,\hat\sigma(k))).
\end{equation}
Taking derivatives with respect to $k$ (in the complementary of the zero measure set where these derivatives are not defined) in both \cref{eqllog} and \cref{eqCharS1}, we find $l'(\delta_{\delta\to k}(k))\frac{d\delta_{\delta\to k}}{dk}(k)=-1$ and $\frac{d\delta_{\delta\to k}}{dk}(k) = n(d_1(k))d_1'(k)$, so in particular
$$d_1'(k)l'(\delta_{\delta\to k}(k))=-\frac{1}{n(d_1(k))}.$$
Since $\delta\to\sigma(\delta)$ lives in $\Sigma_{\delta\to k}$, $l$ is increasing so $d_1$ is decreasing and $\sigma(\delta)$ lives also in $\Sigma_{k\to\delta}$. The function $\delta_{\delta\to k}(k)$, in particular, corresponds to the function $\delta_{k\to\delta}(k)$ defined in $\Sigma_{k\to\delta}$.

\end{proof}

	Given the above proof, we will generally denote the functions
\(\delta_{\delta\to k}(k)\) and \(\delta_{k\to\delta}(k)\) with
\(\delta(k)\), with inverse \(k(\delta)\). We have shown that all smiles
in \(\delta\) which can be transformed unambiguously in a smile in \(k\)
are such that their smile in \(k\) satisfies the Fukasawa first
necessary condition of no arbitrage: \(d_1(k,\hat\sigma(k))\) decreasing
and surjective. Furthermore, in the proof we have shown that a smile in
strike transformed into a smile in delta, can be re-transformed into a
smile in strike and such smile necessarily coincides with the initial
one.

	It is trivial that a smile \(\delta\to\sigma(\delta)\) in
\(\Sigma_{k\to\delta}\) has the property
\[\forall k\in\mathbb{R}\, \exists!\delta(k)\, |\, \delta(k) = N\bigl(d_1\bigl(k,\hat\sigma(k)\bigr)\bigr).\]
In the above proof we showed that such smile has also the property
\[\forall k\in\mathbb{R}\, \exists!\delta(k)\, |\, \delta(k) = N\bigl(d_1\bigl(k,\sigma(\delta(k))\bigr)\bigr).\]

Similarly, from the definition of \(\Sigma_{k\to\delta}\), it follows
that its smiles satisfy
\[\forall \delta\in]0,1[\, \exists!k(\delta)\, |\, \delta = N\bigl(d_1\bigl(k(\delta),\hat\sigma(k(\delta))\bigr)\bigr).\]
In particular, by the definition of the smile
\(\delta\to\sigma(\delta)\), it holds
\(\delta = N(d_1(k(\delta),\sigma(\delta)))\). If there is a second
\(\tilde k\) satisfying \(\delta = N(d_1(\tilde k,\sigma(\delta)))\),
then since \(\Sigma_{k\to\delta}=\Sigma_{\delta\to k}\), it must be
\(\delta=\delta(\tilde k)\), so that \(k(\delta)=\tilde k\). We have
proven that smiles in \(\Sigma_{k\to\delta}\) also satisfy the condition
\begin{equation}\label{eqObsS1S2}
\forall \delta\in]0,1[\, \exists!k(\delta)\, |\, \delta = N\bigl(d_1\bigl(k(\delta),\sigma(\delta)\bigr)\bigr).
\end{equation}

	\hypertarget{qualitative-properties-of-the-smile-from-k-to-delta}{%
\section{\texorpdfstring{Qualitative properties of the smile: from \(k\)
to
\(\delta\)}{Qualitative properties of the smile: from k to \textbackslash delta}}\label{qualitative-properties-of-the-smile-from-k-to-delta}}

	In this section we look at the qualitative properties of the smile in
\(\delta\) resulting from the \(\delta\) transformation of a smile in
\(k\) with decreasing \(k\to d_1(k,\hat\sigma(k))\) function.

	We start with an easy consequence to \cref{PropEqualSets}.

\begin{proposition}

For every smile $\sigma(\delta)\in\Sigma_{\delta\to k}$ it holds
\begin{itemize}
    \item $\sigma'(\delta)\geq 0$ iff $\hat\sigma'(k(\delta))\leq0$ and points of minimum (respectively maximum) $\bar\delta$ for $\sigma(\delta)$ are points of minimum (resp.maximum) $k(\bar\delta)$ for $\hat\sigma(k)$;
    \item $\hat\sigma'(k)\geq 0$ iff $\sigma'(\delta(k))\leq 0$  and points of minimum (respectively maximum) $\bar k$ for $\hat\sigma(k)$ are points of minimum (resp.maximum) $\delta(\bar k)$ for $\sigma(\delta)$.
\end{itemize}

\end{proposition}

	\begin{proof}

Since $\hat\sigma(k(\delta))=\sigma(\delta)$ by definition of $\Sigma_{k\to\delta}$, taking derivatives with respect to $\delta$ implies
\begin{align*}
\sigma'(\delta) &= \hat\sigma'(k(\delta))k'(\delta)\\
&= -\hat\sigma'(k(\delta))l'(\delta)
\end{align*}
where we have used $l(\delta)=-k(\delta)$ as in \cref{eqllog}. Smiles in $\Sigma_{\delta\to k}$ have increasing $l(\delta)$, so the sign of $\sigma'(\delta)$ is opposite to the sign of $\hat\sigma'(k(\delta))$. Furthermore, if $\bar\delta$ is a point of minimum for $\sigma(\delta)$, then for every $\delta$ in a neighborhood of $\bar\delta$, it holds $\sigma(\bar\delta)<\sigma(\delta)$. Using the relation $\hat\sigma(k(\delta))=\sigma(\delta)$, it follows $\hat\sigma(k(\bar\delta))<\hat\sigma(k(\delta))$. Since the function $k(\delta)=d_1^{-1}(N^{-1}(\delta))$ is continuous, then for every $k$ in a neighborhood of $k(\bar\delta)$ it holds $\hat\sigma(k(\bar\delta))<\hat\sigma(k)$, so $k(\bar\delta)$ is a point of minimum for $\hat\sigma(k)$. Similarly for points of maximum.

The proof is similar for the second point, using the relation $\hat\sigma(k)=\sigma(\delta(k))$ from the definition of $\Sigma_{\delta\to k}$, and the fact that $\delta'(k)=\frac{1}{k'(\delta(k))}=-\frac{1}{l'(\delta(k))}$.

\end{proof}

	We have already seen in \cref{PropLee} that under
\cref{remarkHypothesis}, the left and right wing Lee conditions hold,
i.e.~that the limits of \(\frac{\hat\sigma(k)^2T}{k}\) at \(\pm\infty\)
are bounded by \(2\). We now look at what these limits correspond in the
\(\delta\) notation.

	\begin{proposition}\label{PropLimsupLiminf}

Let $\sigma(\delta)\in\Sigma_{k\to\delta}$ and $\hat\sigma(k)$ the corresponding smile in strike. The left wing Lee condition $\frac{\hat\sigma(k)^2T}{k}>-2$ for $k\ll0$ holds iff $\frac{\sigma(\delta)\sqrt T}{N^{-1}(\delta)}<1$ for $\delta$ near $1$. The right wing Lee condition $\frac{\hat\sigma(k)^2T}{k}<2$ for $k\gg0$ holds.

Furthermore,
\begin{align*}
\limsup_{k\to-\infty}\frac{\hat\sigma(k)^2T}{k}=a & & \iff & & \liminf_{\delta\to1}\frac{\sigma(\delta)\sqrt T}{N^{-1}(\delta)}=-\frac{2a}{2-a},
\end{align*}
and
\begin{align*}
\limsup_{k\to\infty}\frac{\hat\sigma(k)^2T}{k}=b & & \iff & & \liminf_{\delta\to0}\frac{\sigma(\delta)\sqrt T}{N^{-1}(\delta)}=-\frac{2b}{2-b}.
\end{align*}

\end{proposition}

\begin{proof}

Since $\sigma(\delta)\in\Sigma_{k\to\delta}$, there exists a smile in strike $\hat\sigma(k)$ with strictly decreasing surjective function $k\to d_1(k,\hat\sigma(k))$. Thanks to the surjectivity of $d_1$, the right wing Lee condition is satisfied as shown in \cref{PropLee}.

For every $k$ it holds $\frac{\hat\sigma(k)^2T}{k} = \frac{\sigma(\delta(k))^2T}{k(\delta(k))}$, where $\delta(k)=N(d_1(k,\hat\sigma(k))$ and $k(\delta)$ is its inverse. Since $\delta(k)$ is surjective decreasing, it goes to $1$ when $k$ goes to $-\infty$ and to $0$ when $k$ goes to $\infty$. Also, since it is monotone continuous, the left wing Lee condition holds iff $\frac{\sigma(\delta)^2T}{k(\delta)}>-2$ for every $\delta$ near $1$, and it holds
\begin{align*}
\limsup_{k\to-\infty}\frac{\hat\sigma(k)^2T}{k} = \limsup_{\delta\to1}\frac{\sigma(\delta)^2T}{k(\delta)}
\end{align*}
and similarly for $k\to\infty$. From \cref{eqllog}, we can substitute $k(\delta)$ with $-l(\delta)$, which is defined in \cref{eqDefL}, so that we are now studying the quantity $-\frac{\sigma(\delta)\sqrt T}{N^{-1}(\delta)-\frac{\sigma(\delta)\sqrt T}2}$. Since $l(\delta)$ is increasing surjective, the denominator is negative for small $\delta$ and positive for large $\delta$. Then, it is easy to see that the left wing Lee condition holds iff $\frac{\sigma(\delta)\sqrt T}{N^{-1}(\delta)}<1$.

The argument of the limit becomes $\bigl(\frac12-\frac{N^{-1}(\delta)}{\sigma(\delta)\sqrt T}\bigr)^{-1}$, so that the two limits superior are equal to
$$\Bigl(\frac12 - \limsup_{\delta}\frac{N^{-1}(\delta)}{\sigma(\delta)\sqrt T}\Bigr)^{-1}.$$
This quantity is equal to $c$ (either $a$ for $\delta\to1$ or $b$ for $\delta\to0$) iff $\limsup_{\delta}\frac{N^{-1}(\delta)}{\sigma(\delta)\sqrt T}=\frac{c-2}{2c}$ and the conclusion follows. The reasoning still holds for $c=0$.

\end{proof}

	From the above proposition it follows that Lee conditions applied to the
total variance limits (i.e.~\(a\in[-2,0]\) and \(b\in[0,2]\)), translate
in the delta notation into the requirement that the limit at \(0\) of
\(\frac{\sigma(\delta)\sqrt T}{N^{-1}(\delta)}\) is negative while the
limit at \(1\) is positive and smaller than \(1\). The first condition,
i.e.~the right wing Lee condition, is automatically granted by the sign
of the function \(N^{-1}(\delta)\). This is what we expected using
\cref{PropLee} since a smile in delta which can be transformed in a
smile in strike has surjective function \(k\to d_1(k,\hat\sigma(k))\).
The second condition, i.e.~the left wing Lee condition, is not
automatically granted. However, it is granted under
\cref{remarkHypothesis} as shown in \cref{PropLee}.

	\begin{remark}

In section 4.1 of \cite{schluter2009tail}, it is shown that the high Gaussian quantile can be asymptotically wrote as $N^{-1}(\delta) = \sqrt{-2\log(1-\delta)} + \mathcal o(N^{-1}(\delta))$. With a similar reasoning, it is easy to prove that for low Gaussian quantiles it holds $N^{-1}(\delta) = -\sqrt{-2\log(\delta)} + \mathcal o(N^{-1}(\delta))$. Then, the limit inferior in \cref{PropLimsupLiminf} can be substituted with
\begin{align*}
\liminf_{\delta\to1}\frac{\sigma(\delta)\sqrt T}{\sqrt{-2\log(1-\delta)}} & & \text{and} & & \liminf_{\delta\to0}\frac{\sigma(\delta)\sqrt T}{-\sqrt{-2\log(\delta)}}.
\end{align*}

\end{remark}

	We now look at the expansion of a smile in delta around the ATM point
for a given expansion of the corresponding smile in strike.

\begin{proposition}

Let $\sigma(\delta)\in\Sigma_{k\to\delta}$ and $\hat\sigma(k)$ the corresponding smile in strike. If
\begin{equation*}
\hat\sigma(k)\sqrt{T} = a_0 + a_1 k + a_2 k^2 + \mathcal o(k^2)
\end{equation*}
then
\begin{align*}
\sigma(\delta)\sqrt{T} =\,& a_0 - \frac{2a_0 a_1}{n\bigl(\frac{a_0}2\bigr)(2-a_0a_1)} (\delta-\delta(0))\, +\\
&+ \frac{a_0\bigl(8(a_0a_2+2a_1^2) - a_0a_1(2-a_0a_1)^2\bigr)}{n\bigl(\frac{a_0}2\bigr)^2(2-a_0a_1)^3} (\delta-\delta(0))^2\, +\\
&+ \mathcal o((\delta-\delta(0))^2)
\end{align*}
where $\delta(0)=N\bigl(\frac{a_0}{2}\bigr)$ is the delta ATM point.

\end{proposition}

\begin{proof}

The delta ATM point is $\delta(0) = N(d_1(0,\hat\sigma(0)) = N\bigl(\frac{a_0}2\bigr)$. Observe that from \cref{eqllog}, it holds $l(\delta)=-k(\delta)$. We will use the following relations:
\begin{align}
& \sigma(\delta) = \hat\sigma(k(\delta)) \label{eqFromSigmaDeltaToK}\\
& \sigma'(\delta) = -\hat\sigma'(k(\delta))l'(\delta) \label{eqFromSigmaDeltaToKDeriv}\\
& \sigma''(\delta) = \hat\sigma''(k(\delta))l'(\delta)^2 - \hat\sigma'(k(\delta))l''(\delta). \label{eqFromSigmaDeltaToKDeriv2}
\end{align}
The first and second derivatives of $l(\delta)$ can be computed from the definition in \cref{eqDefL} as
\begin{align*}
l'(\delta) &= \biggl(\frac{1}{n(N^{-1}(\delta))}-\frac{\sigma'(\delta)\sqrt T}2\biggr)\sigma(\delta)\sqrt T + l(\delta)\frac{\sigma'(\delta)}{\sigma(\delta)}\\
l''(\delta) &= \biggl(\frac{N^{-1}(\delta)}{n(N^{-1}(\delta))^2}-\frac{\sigma''(\delta)\sqrt T}2\biggr)\sigma(\delta)\sqrt T + \biggl(2l'(\delta)-l(\delta)\frac{\sigma'(\delta)}{\sigma(\delta)}\biggr)\frac{\sigma'(\delta)}{\sigma(\delta)} + l(\delta)\frac{d}{d\delta}\frac{\sigma'(\delta)}{\sigma(\delta)}.
\end{align*}

From \cref{eqFromSigmaDeltaToK}, it follows that the constant coefficient of the expansion of $\sigma(\delta)\sqrt T$ is $a_0$.

Observe that $l(\delta(0))=-k(\delta(0))=0$, so that $l'(\delta(0))$ has only one term. For the first order coefficient, from \cref{eqFromSigmaDeltaToKDeriv} we obtain $\sigma'(\delta(0))\sqrt T = -a_1l'(\delta(0))$. Substituting in the expression for the derivative of $l(\delta)$ and solving for $l'(\delta(0))$, it follows
$$l'(\delta(0)) = \frac{2a_0}{n\bigl(\frac{a_0}2\bigr)(2-a_0a_1)},$$
so that we find the first order coefficient of the expansion of $\sigma(\delta)\sqrt T$.

Finally, the second order coefficient can be found from the expression for the second derivative of $l(\delta)$ evaluated in $\delta(0)$. As in the previous steps, substituting $\sigma''(\delta(0))\sqrt T$ with $a_2l'(\delta(0))^2 - a_1l''(\delta(0))$ from \cref{eqFromSigmaDeltaToKDeriv2} and solving for $l''(\delta(0))$, we find
$$l''(\delta(0)) = \frac{a_0\bigl(a_0(2-a_0a_1)^2-4(a_0^2a_2+4a_1)\bigr)}{n\bigl(\frac{a_0}2\bigr)^2(2-a_0a_1)^3}$$
and from this the expression for the second order coefficient of the expansion of $\sigma(\delta)\sqrt T$.

\end{proof}

	\hypertarget{symmetric-transformation}{%
\section{Symmetric transformation}\label{symmetric-transformation}}

	Consider the smile inversion
\(k\to\hat\sigma(-k):=\hat{\bar\sigma}(k)\). Arbitrage-free smiles are
such that their inverse smile is still arbitrage-free. The
parametrisation of the set \(\Sigma_\text{A}\) is hard. Indeed, twice
derivable functions \(\sigma(\delta)\) belong to \(\Sigma_\text{A}\) iff
they satisfy the delta version of the requirement of positivity of
\cref{eqButtArbitrage}. Observe that such condition can be written only
for the subset of \(\Sigma_\text{A}\) of twice derivable functions since
it involves second derivatives of the smile.

Nevertheless, the subset of \(\Sigma_{\delta\to k}\) which is closed
under symmetry can be parametrized, and this could make a step forward
to the search of arbitrage-free smiles in delta. We define such set
\(\Sigma_\text{WA}\) since it is the set of smiles satisfying the Weak
Arbitrage-free conditions of no arbitrage of monotonicity of \(d_1(k)\)
and \(d_2(k)\) found by Fukasawa (plus the surjectivity of such
functions). The requirement defining \(\Sigma_{k\to\delta}\) is that
there exists a smile \(k\to\hat\sigma(k)\) such that the function
\(d_1(k)=d_1(k,\hat\sigma(k))\) is decreasing and surjective. For the
inverse smile, we are asking that the function
\(\bar d_1(k)=d_1(k,\hat{\bar\sigma}(k))\) is decreasing and surjective.
It is easy to see that \(\bar d_1(k)=-d_2(-k)\).

Taking the function \(\sigma(\delta)=\hat\sigma(k(\delta))\) where
\(k(\delta)=d_1^{-1}(N^{-1}(\delta))\), the requirement that \(d_2(k)\)
is decreasing and surjective corresponds to the requirement that
\(\delta\to d_2(k(\delta))\) is increasing and surjective. It holds
\(d_2(k)=d_1(k)-\hat\sigma(k)\sqrt T\), so in the delta notation the
requirement is that the function
\[m(\delta) = N^{-1}(\delta)-\sigma(\delta)\sqrt{T}\] is increasing and
surjective.

We can then define the subset of \(\Sigma_{k\to\delta}\) closed for
smile inversion as \begin{equation}\label{eqSigmaWA}
\begin{aligned}
\Sigma_\text{WA} := \bigl\{\delta\to\sigma(\delta)\in \mathcal D(]0,1[,\mathbb{R}^+)\, \bigl|\, \exists k\to\hat\sigma(k)\in \mathcal D(\mathbb{R},\mathbb{R}^+)\ \text{s.t.}&\\
k\to d_1(k,\hat\sigma(k))\ \text{strictly decreasing surjective},&\\
k\to d_2(k,\hat\sigma(k))\ \text{strictly decreasing surjective},&\\
\hat\sigma(k(\delta))=\sigma(\delta)&\bigr\}.
\end{aligned}
\end{equation}

It is easy to see that if \(\sigma(\delta)\) belongs to
\(\Sigma_\text{WA}\), then it is possible to define the smile
\(\bar\sigma(\delta):=\hat{\bar\sigma}(\bar k(\delta))\) where
\(\bar k(\delta) = \bar{d}_1^{-1}(N^{-1}(\delta)) = -d_2^{-1}(-N^{-1}(\delta))\),
so that \(\bar\sigma(\delta)=\hat\sigma(d_2^{-1}(-N^{-1}(\delta)))\).
Since \(d_2(k)\) is strictly decreasing, it is a.s. differentiable and
so is its inverse. As a consequence, because \(\hat\sigma(k)\) admits
left/right derivative, also \(\bar\sigma(\delta)\) admits them and it
belongs to \(\Sigma_\text{WA}\).

	In order to parametrize \(\Sigma_\text{WA}\) we look at the problem of
calibration of a smile in \(\delta\).

	\hypertarget{calibration-of-a-smile-in-delta-in-sigma_deltato-k}{%
\subsection{\texorpdfstring{Calibration of a smile in delta in
\(\Sigma_{\delta\to k}\)}{Calibration of a smile in delta in \textbackslash Sigma\_\{\textbackslash delta\textbackslash to k\}}}\label{calibration-of-a-smile-in-delta-in-sigma_deltato-k}}

	We have characterized the set \(\Sigma_{\delta\to k}\) of admissible
functions \(\delta\to\sigma(\delta)\) such that it is possible to
recover a function \(k\to\hat\sigma(k)\) where
\(\hat\sigma(k)=\sigma(\delta(k))\) and
\(\delta(k) = N(d_1(k, \sigma(\delta(k))))\). Also, we have proven that
a function \(\delta\to\sigma(\delta)\) belongs to
\(\Sigma_{\delta\to k}\) iff it is possible to recover a function
\(k\to\hat\sigma(k)\) satisfying the first Fukasawa necessary condition.
As a consequence, if a function \(\delta\to\sigma(\delta)\) belongs to
\(\Sigma_{\delta\to k}\), then the function \(k\to\delta(k)\) is
bijective. The choice of setting
\(\sigma(\delta)=\hat\sigma(k(\delta))\) is in fact a requirement, since
\(\hat\sigma(k(\delta))=\sigma(\delta(k(\delta)))=\sigma(\delta)\), so
that from a smile in delta living in \(\Sigma_{\delta\to k}\), it can be
possible to pass through a smile in \(k\) and to go back to the same
smile in \(\delta\).

	Suppose we want to calibrate a smile in \(\delta\) which can be
converted to a smile in \(k\) and vice-versa. We could go through the
following steps:

\begin{enumerate}
\def\labelenumi{\arabic{enumi}.}
\tightlist
\item
  consider the market discrete pillars \(\{k_i,\sigma_i\}_i\);
\item
  convert them to the pillars \(\{\delta_i,\sigma_i\}_i\) by defining
  \(\delta_i= N(d_1(k_i,\sigma_i))\);
\item
  compute the pillars \(\{\delta_i,l_i\}\) with \(l_i=-k_i\);
\item
  interpolate/extrapolate an increasing and surjective function
  \(\delta\to l(\delta)\), given pillars in point 3.
\end{enumerate}

The last natural point would be to recover a function
\(\delta\to\sigma(\delta)\) by the calibrated function \(l(\delta)\).
However, this reduces to solve the equation
\[\frac{\sigma^2T}{2} - N^{-1}(\delta)\sqrt{T}\sigma + l(\delta) =0,\]
which could have no, one or two solutions. In the following section we
study the problem of existence and uniqueness of the solution
\(\sigma\).

	\hypertarget{conditions-on-l-for-the-existence-of-sigma}{%
\subsubsection{\texorpdfstring{Conditions on \(l\) for the existence of
\(\sigma\)}{Conditions on l for the existence of \textbackslash sigma}}\label{conditions-on-l-for-the-existence-of-sigma}}

	We now look at the conditions on a given increasing and surjective
function \(\delta\to l(\delta)\) in order to have that
\begin{equation}\label{eqSigmaFromL}
\forall \delta\, \exists! \sigma(\delta) \, |\, \frac{\sigma^2T}{2} - N^{-1}(\delta)\sqrt{T}\sigma + l(\delta) =0 \; \text{with} \; \sigma=\sigma(\delta).
\end{equation}

	\begin{proposition}\label{propLUnique}

Let $l:]0,1[\to\mathbb R$ an increasing and surjective function. The equation
\begin{equation}\label{eqLSigma}
\frac{\sigma^2T}{2} - N^{-1}(\delta)\sqrt{T}\sigma + l(\delta) =0
\end{equation}
has at least one solution $\sigma$ for every $\delta\in]0,1[$ iff $l\bigl(\frac12\bigr)<0$ and $l(\delta)\leq\frac{N^{-1}(\delta)^2}{2}$ for all $\delta>\frac12$. The continuous solution $\delta\to\sigma(\delta)$ is
\begin{equation}\label{eqSigmaP}
N^{-1}(\delta) + \sqrt{N^{-1}(\delta)^2-2l(\delta)}
\end{equation}
for $\delta\leq \frac12$ and it could switch to
\begin{equation}\label{eqSigmaM}
N^{-1}(\delta) - \sqrt{N^{-1}(\delta)^2-2l(\delta)}
\end{equation}
at every $\tilde\delta>\frac12$ such that $l(\tilde\delta)=\frac{N^{-1}(\tilde\delta)^2}{2}$.

\end{proposition}

	\begin{proof}

For a fixed $\delta\in]0,1[$, the two admissible sigma solutions to \cref{eqLSigma} are
\begin{equation*}
\sigma_\pm(\delta) = \frac{N^{-1}(\delta)}{\sqrt T} \pm \frac{\sqrt{N^{-1}(\delta)^2-2l(\delta)}}{\sqrt{T}}.
\end{equation*}
For the existence, the delta of the equation must be non-negative and at least one of the two solutions must be positive, so
\begin{align*}
&N^{-1}(\delta)^2-2l(\delta)\geq0,\\
&N^{-1}(\delta) \pm \sqrt{N^{-1}(\delta)^2-2l(\delta)}>0,
\end{align*}
where the sign depends on the chosen solution.

Firstly, when $\delta\leq\frac12$, the quantity $N^{-1}(\delta)$ is non-positive and the $-$ solution is negative, so it can be discarded. Instead, the $+$ solution is well-defined and positive iff $l(\delta)<0$. Since $l$ is increasing, the latter condition is equivalent to $l\bigl(\frac12\bigr)<0.$

When $\delta>\frac12$, then $N^{-1}(\delta)$ is positive and both the solutions could be valid. Under the requirement $l(\delta)\leq\frac{N^{-1}(\delta)^2}{2}$, the $+$ solution is always positive while the $-$ solution becomes positive when $l$ becomes positive, and it will stay positive since $l$ is increasing.

The possibility to pass from the $+$ to the $-$ solution at a point $\tilde\delta$ must be guaranteed by continuity of the volatility on $\tilde\delta$. Then, it should hold $\sigma_+(\tilde\delta)=\sigma_-(\tilde\delta)$, or $2l(\tilde\delta)=N^{-1}(\tilde\delta)^2$ and $\sigma(\tilde\delta)\sqrt{T}=N^{-1}(\tilde\delta)$. Rewriting equation \cref{eqLSigma} as
$$l(\delta)=\Bigl(N^{-1}(\delta)-\frac12 \sigma(\delta) \sqrt{T}\Bigr)\sigma(\delta) \sqrt{T}$$
and evaluating in $\tilde\delta$, one finds $l(\tilde\delta)=\frac{\sigma(\tilde\delta)^2T}2$. If there exists such a point $\tilde\delta$, then either the solution keeps being $\sigma_+$ after $\tilde\delta$, or it switches to $\sigma_-$. If there is no point $\tilde\delta$, the solution remains $\sigma_+$.

\end{proof}

	As an immediate consequence of \cref{propLUnique}, under the
requirements \(l\bigl(\frac12\bigr)<0\) and
\(l(\delta)\leq\frac{N^{-1}(\delta)^2}{2}\), if there are no
\(\tilde\delta\) such that
\(l(\tilde\delta)=\frac{N^{-1}(\tilde\delta)^2}{2}\), the \(\sigma\)
solution is unique and it coincides with \cref{eqSigmaP}. If there are
one or more points \(\tilde\delta\), the uniqueness is not guaranteed
since the \(\sigma\) solution could either switch between
\cref{eqSigmaP} and \cref{eqSigmaM} or not do the switch. For the
uniqueness of the solution, more requirements on the solution itself are
needed. We show in the next section that if we require the solution
\(\sigma\) to lie in \(\Sigma_{\text{WA}}\), then the uniqueness is
satisfied.

	\hypertarget{parametrization-of-sigma_textwa}{%
\subsection{\texorpdfstring{Parametrization of
\(\Sigma_{\text{WA}}\)}{Parametrization of \textbackslash Sigma\_\{\textbackslash text\{WA\}\}}}\label{parametrization-of-sigma_textwa}}

	Arbitrage-free smiles in log-forward moneyness satisfy some necessary
conditions. We have already seen the Fukasawa first necessary condition
proven in \cite{fukasawa2012normalizing}, Theorem 2.8 which states that
the function \(k\to d_1(k,\hat\sigma(k))\) is decreasing. Under
hypothesis in \cref{remarkHypothesis}, \(d_1(k)\) is also surjective.
This implies that \(d_1(k)\) must be negative for \(k\) large enough,
and this in turn gives the right wing Lee condition
\(\hat\sigma(k)\sqrt{T}<\sqrt{2k}\) for \(k\gg0\) (\cite{lee2004moment},
Lemma 3.1).

In the same theorem, Fukasawa proves that also the function
\(k\to d_2(k,\hat\sigma(k))\) is decreasing. Again,
\cref{remarkHypothesis} assumes also its surjectivity. This implies that
\(d_2(k)\) must be positive for \(k\) small enough and that the left
wing Lee condition holds: \(\hat\sigma(k)\sqrt{T}<\sqrt{-2k}\) for
\(k\ll0\) (\cite{lee2004moment}, Lemma 3.3). Lee shows that the latter
condition holds for every arbitrage-free smile iff
\(P(S_T=0)<\frac{1}2\), which is indeed our case since we suppose no
mass in \(0\).

	We show in \cref{lemmaS1LSolSigma} that smiles satisfying the former
Fukasawa necessary condition guarantee the existence of a \(\sigma\)
solution to \cref{eqLSigma} and in \cref{lemmaT1SigmaSol} that smiles
satisfying also the latter Fukasawa necessary condition guarantee the
existence and uniqueness of a \(\sigma\) solution.

	\begin{lemma}\label{lemmaS1LSolSigma}

For every $\delta\to\sigma(\delta)\in \Sigma_{\delta\to k}$, the function $l$ satisfies $l\bigl(\frac12\bigr)<0$ and $l(\delta)\leq\frac{N^{-1}(\delta)^2}{2}$ for all $\delta>\frac12$.

\end{lemma}

	\begin{proof}

As seen in \cref{eqObsS1S2}, for any $\delta$ there is a (unique) $k(\delta)$ such that $\frac{\sigma(\delta)^2T}{2} - N^{-1}(\delta)\sqrt{T}\sigma(\delta) - k(\delta) =0$. Then, $\sigma=\sigma(\delta)$ is a solution of \cref{eqLSigma} and the conclusion follows from \cref{propLUnique}.

The statement can also be proven by hand. Indeed, $l\bigl(\frac12\bigr)$ is equal to $-\frac{\sigma(\frac12)^2T}{2}$, which is always negative, and the second condition reads
$$\Bigl(N^{-1}(\delta)-\frac12 \sigma(\delta) \sqrt{T}\Bigr)\sigma(\delta) \sqrt{T} \leq \frac{N^{-1}(\delta)^2}{2},$$
which simplifying becomes $(N^{-1}(\delta)-\sigma(\delta)\sqrt{T})^2\geq0$, which is always verified.

\end{proof}

	In the case of \(\Sigma_\text{WA}\), \cref{lemmaS1LSolSigma} can be
further developed showing that both the existence and the uniqueness of
a solution \(\sigma\) to \cref{eqLSigma} hold.

	\begin{lemma}\label{lemmaT1SigmaSol}

For every $\delta\to\sigma(\delta)\in \Sigma_\text{WA}$, the function $l$ satisfies $l\bigl(\frac12\bigr)<0$, $l(\delta)\leq\frac{N^{-1}(\delta)^2}{2}$ for all $\delta>\frac12$ and there exists a unique $\tilde\delta$ such that $l(\tilde\delta)=\frac{N^{-1}(\tilde\delta)^2}{2}$. Furthermore
\begin{equation}\label{eqSigmaT1}
\sigma(\delta)\sqrt{T} =
\begin{cases}
N^{-1}(\delta) + \sqrt{N^{-1}(\delta)^2-2l(\delta)} & \text{if $\delta\leq\tilde\delta$},\\
N^{-1}(\delta) - \sqrt{N^{-1}(\delta)^2-2l(\delta)} & \text{if $\delta>\tilde\delta$.}\\
\end{cases}
\end{equation}

\end{lemma}

	\begin{proof}

Since $\Sigma_\text{WA}$ is a subset of $\Sigma_{\delta\to k}$, the conditions $l\bigl(\frac12\bigr)<0$ and $l(\delta)\leq\frac{N^{-1}(\delta)^2}{2}$ for all $\delta>\frac12$ are satisfied by \cref{lemmaS1LSolSigma}.

Since $d_2(k)$ is a decreasing and surjective function, the left wing Lee condition $\hat\sigma(k)\sqrt{T}<\sqrt{-2k}$ for $k\ll0$ holds. Since $l(\delta(k))=-k$, the condition becomes $\sigma(\delta(k))\sqrt{T}<\sqrt{2l(\delta(k))}$ for $k$ negative enough. Given the monotonicity of the function $\delta(k)$, this is equivalent to requiring $\sigma(\delta)\sqrt{T}<\sqrt{2l(\delta)}$ for $\delta$ large enough, in particular $\delta\gg\frac12$. The $\sigma$ solution of \cref{eqLSigma} can be either of the form $\sigma_+$ \cref{eqSigmaP} or $\sigma_-$ \cref{eqSigmaM}. Computing the square of $\sigma_\pm$, the Lee condition is
$$N^{-1}(\delta)\bigl(N^{-1}(\delta)\pm\sqrt{N^{-1}(\delta)^2-2l(\delta)}\bigr)<2l(\delta).$$
Since $\delta>\frac12$, the quantity $N^{-1}(\delta)$ is positive and dividing we get
$$\pm\sqrt{N^{-1}(\delta)^2-2l(\delta)}<\frac{2l(\delta)}{N^{-1}(\delta)}-N^{-1}(\delta).$$
The right hand side is positive iff $2l(\delta)>N^{-1}(\delta)^2$, which cannot hold. Then the $\sigma_+$ solution does not satisfy the left wing Lee condition for small $k$. On the other hand, the $\sigma_-$ solution always satisfies it since the above inequality holds true iff $2l(\delta)<N^{-1}(\delta)^2$.

Then, the $\sigma$ solution is equal to $\sigma_+$ for $\delta$ smaller than a certain $\tilde\delta$ such that $l(\tilde\delta)=\frac{N^{-1}(\tilde\delta)^2}{2}$ and then it switches to $\sigma_-$. The uniqueness of such point $\tilde\delta$ follows from the monotonicity of $d_2(k)$. Indeed, if two points $\tilde\delta$ and $\hat\delta$ satisfy $l(\delta)=\frac{N^{-1}(\delta)^2}{2}$, then they also satisfy $\sigma(\delta)\sqrt{T}= N^{-1}(\delta)$. In the log-forward moneyness notation, there exist $\tilde k=k(\tilde\delta)$ and $\hat k=k(\hat\delta)$ which satisfy $\hat\sigma(k)\sqrt{T} = d_1(k)$, or $d_2(k)=0$. Since $d_2$ is one-to-one, $\tilde k=\hat k$ and $\tilde\delta=\hat\delta$.

\end{proof}

	Thanks to the above result, it is possible to parametrize the smile
\(\sigma\) living in \(\Sigma_\text{WA}\) using a parametrization of the
function \(l\), which has to be increasing and surjective and has to
satisfy the conditions of existence and uniqueness of
\cref{lemmaT1SigmaSol}.

We finally state how to parameterize the set \(\Sigma_\text{WA}\). We
show the result in case of a.s. differentiable implied volatilities.

	\begin{theorem}[Parametrization of weak arbitrage-free smiles]\label{theoParamT1}

Any a.s. differentiable $\sigma(\delta)\in \Sigma_\text{WA}$ can be parameterized as
\begin{equation}\label{eqSigmaT1}
\sigma(\delta)\sqrt{T} =
\begin{cases}
N^{-1}(\delta) + \sqrt{N^{-1}(\delta)^2 + 2\bigl(\int_\delta^{\frac12}\lambda(x)\,dx + \int_{\frac12}^{\tilde\delta}\mu(x)\,dx\bigr)} & \text{if $\delta\leq\frac12$},\\
N^{-1}(\delta) + \sqrt{2\int_\delta^{\tilde\delta}\mu(x)\,dx} & \text{if $\frac12<\delta\leq\tilde\delta$},\\
N^{-1}(\delta) - \sqrt{2\int_{\tilde\delta}^\delta\frac{N^{-1}(x)}{n(N^{-1}(x))}\alpha(x)\,dx} & \text{if $\delta>\tilde\delta$}.\\
\end{cases}
\end{equation}
where
* $\tilde\delta\in\bigl]\frac12,1\bigr[$;
* $\lambda$ is an a.s. continuous positive function defined on $\bigl]0,\frac12\bigr]$ such that $\int_0^{\frac12}\lambda(x)\,dx = \infty$;
* $\mu$ is an a.s. continuous positive function defined on $\bigl[\frac12,\tilde\delta\bigr[$;
* $\alpha$ is a. as. continuous function defined on $]\tilde\delta,1[$ such that $\alpha(\delta)\in]0,1[$ a.s., $\int_{\tilde\delta}^{1}\frac{N^{-1}(x)}{n(N^{-1}(x))}\alpha(x)\,dx = \infty$, and $\lim_{\delta\to1^-}\bigl(\frac{N^{-1}(\delta)^2}2 - \int_{\tilde\delta}^{\delta}\frac{N^{-1}(x)}{n(N^{-1}(x))}\alpha(x)\,dx\bigr) = \infty$.

\end{theorem}

	\begin{proof}

Let $\sigma(\delta)\in \Sigma_\text{WA}$ and let $\sigma(\delta)$ be a.s. differentiable. Then for \cref{lemmaT1SigmaSol}, $\sigma(\delta)$ has the form in \cref{eqSigmaT1}. The function $l(\delta)$ is increasing and surjective and it satisfies $l\bigl(\frac12\bigr)<0$, $l(\delta)\leq\frac{N^{-1}(\delta)^2}{2}$ for all $\delta>\frac12$, and there exists a unique $\tilde\delta$ such that $l(\tilde\delta)=\frac{N^{-1}(\tilde\delta)^2}{2}$. Since $d_2(k)$ is decreasing and surjective, the function $m(\delta)=N^{-1}(\delta)-\sigma(\delta)\sqrt T$ is increasing and surjective. Substituting with the expression for $\sigma(\delta)$, $m(\delta)$ can be re-written as $\mp\sqrt{N^{-1}(\delta)^2-2l(\delta)}$, where the sign is negative for $\delta\leq\tilde\delta$ and positive otherwise. Its derivative, when it is defined, is $\mp\frac{\frac{N^{-1}(\delta)}{n(N^{-1}(\delta))}-l'(\delta)}{\sqrt{N^{-1}(\delta)^2-2l(\delta)}}$, and it is positive. Equivalently, the derivative of $l$ satisfies
\begin{equation}\label{eqMMonotone}
\begin{aligned}
&l'(\delta)>\frac{N^{-1}(\delta)}{n(N^{-1}(\delta))} \, \text{a.s.} & \text{if $\delta<\tilde\delta$},\\
&l'(\delta)<\frac{N^{-1}(\delta)}{n(N^{-1}(\delta))} \, \text{a.s.} & \text{if $\delta>\tilde\delta$}
\end{aligned}
\end{equation}
The first inequality is weaker than $l'(\delta)>0$ a.s. if $\delta<\frac12$. Consider then $\delta\in\bigl[\frac12,\tilde\delta\bigr[$. The first inequality implies that there is a positive and a.s. continuous function $\mu(\delta)$ on $\bigl[\frac12,\tilde\delta\bigr[$, such that $l'(\delta)=\frac{N^{-1}(\delta)}{n(N^{-1}(\delta))} + \mu(\delta)$. Taking the integral from $\delta$ to $\tilde\delta$ results into $l(\delta)=\frac{N^{-1}(\delta)^2}{2} - \int_\delta^{\tilde\delta}\mu(x)\,dx$. If $\delta\leq\frac12$, the fact that $l'(\delta)$ is positive can be written as $l'(\delta)=\lambda(\delta)$ where $\lambda(\delta)$ is an a.s. continuous positive function defined on $\bigl]0,\frac12\bigr]$. Taking the integral between $\delta$ and $\frac12$ implies $l(\delta)=l\bigl(\frac12\bigr)-\int_\delta^{\frac12}\lambda(x)\,dx$. Substituting with the value of $l\bigl(\frac12\bigr)$ in the expression with $\mu$, it holds $l(\delta)= -\int_\delta^{\frac12}\lambda(x)\,dx - \int_{\frac12}^{\tilde\delta}\mu(x)\,dx$. Since $l(\delta)$ is surjective and $l\bigl(\frac12\bigr)$ is finite, then $\int_0^{\frac12}\lambda(x)\,dx=\infty$. Similarly, for $\delta>\tilde\delta$, the property of the derivative of $l(\delta)$ implies $l(\delta)=\frac{N^{-1}(\delta)^2}{2} - \int^\delta_{\tilde\delta}\eta(x)\,dx$ for an a.s. continuous positive function $\eta(\delta)$ defined on $]\tilde\delta,1[$. Also, since $l'(\delta)>0$, $\eta(\delta)$ is smaller than $\frac{N^{-1}(\delta)}{n(N^{-1}(\delta))}$, so $\eta(\delta)=\alpha(\delta)\frac{N^{-1}(\delta)}{n(N^{-1}(\delta))}$ where $\alpha(\delta)$ is a function strictly bounded between $0$ and $1$ and a.s. continuous on $]\tilde\delta,1[$. Furthermore, $l(1)=\infty$, then $\lim_{\delta\to1^-}\bigl(\frac{N^{-1}(\delta)^2}2 - \int_{\tilde\delta}^{\delta}\frac{N^{-1}(x)}{n(N^{-1}(x))}\alpha(x)\,dx\bigr) = \infty$. In order to have $m(1)=\infty$, it must hold $\int_{\tilde\delta}^{1}\frac{N^{-1}(x)}{n(N^{-1}(x))}\alpha(x)\,dx=\infty$.

On the other hand, if a function $\sigma(\delta)$ has the form in \cref{eqSigmaT1}, the function $l(\delta)$ is
\begin{equation*}
l(\delta)=
\begin{cases}
-\int_\delta^{\frac12}\lambda(x)\,dx - \int_{\frac12}^{\tilde\delta}\mu(x)\,dx & \text{if $\delta\leq\frac12$},\\
\frac{N^{-1}(\delta)^2}{2}-\int_\delta^{\tilde\delta}\mu(x)\,dx & \text{if $\frac12<\delta\leq\tilde\delta$},\\
\frac{N^{-1}(\delta)^2}{2}-\int_{\tilde\delta}^\delta\frac{N^{-1}(x)}{n(N^{-1}(x))}\alpha(x)\,dx & \text{if $\delta>\tilde\delta$}.\\
\end{cases}
\end{equation*}
Given the hypothesis on the parameters, it is easy to show that $l(\delta)$ is a.s. differentiable and $l'(\delta)>0$ for every $\delta$ where the derivative is defined. Also, $l(0)=-\int_0^{\frac12}\lambda(x)\,dx =-\infty$ and $l(1)= \infty$. So far, we have proven $\sigma(\delta)\in \Sigma_{\delta\to k}$. The only requirement left is $m(\delta)$ increasing and surjective. The monotonicity holds since inequalities in \cref{eqMMonotone} are verified. For the surjectivity, $m(0)=-\sqrt{N^{-1}(0)^2 + 2\bigl(\int_0^{\frac12}\lambda(x)\,dx + \int_{\frac12}^{\tilde\delta}\mu(x)\,dx\bigr)}$ which is $-\infty$, while $m(1)$ is equal to $\sqrt{2\int_{\tilde\delta}^{1}\frac{N^{-1}(x)}{n(N^{-1}(x))}\alpha(x)\,dx}$ which diverges.

\end{proof}

	\hypertarget{study-of-the-parameters}{%
\subsubsection{Study of the parameters}\label{study-of-the-parameters}}

	\hypertarget{relations-with-l-and-m}{%
\paragraph{\texorpdfstring{Relations with \(l\) and
\(m\)}{Relations with l and m}}\label{relations-with-l-and-m}}

	In this paragraph we have a look at parameters \(\tilde\delta\),
\(\lambda\), \(\mu\) and \(\alpha\) and study their relation with the
two functions \(l(\delta)\) and \(m(\delta)\).

Given a function \(\delta\to\sigma(\delta)\) in \(\Sigma_\text{WA}\),
the point \(\tilde\delta\) is the only solution (which will be
automatically greater than \(\frac12\)) to
\(l(\delta) = \frac{N^{-1}(\delta)^2}2\). Equivalently, the point
\(\tilde\delta\) is the only solution to \(m(\delta)=0\).

The function \(\mu\) can be recovered from
\[\sigma(\delta)\sqrt{T} = N^{-1}(\delta) + \sqrt{2\int_\delta^{\tilde\delta}\mu(x)\,dx}\]
for \(\delta\in\bigl[\frac12,\tilde\delta\bigr[\). In particular,
\(\int_\delta^{\tilde\delta}\mu(x)\,dx = \frac{(\sigma(\delta)\sqrt{T}-N^{-1}(\delta))^2}2 = \frac{m(\delta)^2}{2}\),
and deriving one finds \[\mu(\delta)=-m(\delta)m'(\delta).\]

In the proof of \cref{theoParamT1}, we showed
\(\lambda(\delta) = l'(\delta)\) for \(\delta\leq\frac12\).

Finally, consider \(\delta>\tilde\delta\). Then
\[\sigma(\delta)\sqrt{T} = N^{-1}(\delta) - \sqrt{2\int_{\tilde\delta}^\delta\frac{N^{-1}(x)}{n(N^{-1}(x))}\alpha(x)\,dx}\]
and similarly as before
\[\alpha(\delta) = \frac{n(N^{-1}(\delta))}{N^{-1}(\delta)}m(\delta)m'(\delta).\]

	\hypertarget{requirements-on-parameters}{%
\paragraph{Requirements on
parameters}\label{requirements-on-parameters}}

	In \cref{theoParamT1}, the positivity of parameter \(\lambda\) and the
requirement \(\alpha(\delta)<1\) are directly linked to the fact that
the function \(l(\delta)\) must be increasing. The positivity of \(\mu\)
and \(\alpha\) is instead connected with the monotonicity of the
function \(m(\delta)\).

The requirement \(\int_0^{\frac12}\lambda(x)\,dx = \infty\) comes from
the fact that \(l(0)=-\infty\) and it also implies \(m(0)=-\infty\). The
requirement
\(\int_{\tilde\delta}^{1}\frac{N^{-1}(x)}{n(N^{-1}(x))}\alpha(x)\,dx = \infty\)
originates from \(m(1)=\infty\).

	The last and most awkward requirement arises to satisfy \(l(1)=\infty\).
Indeed, in order to have the sufficient and necessary condition, the
requirement \begin{equation}\label{eqAlphaLim}
\lim_{\delta\to1^-}\Bigl(\frac{N^{-1}(\delta)^2}2 - \int_{\tilde\delta}^{\delta}\frac{N^{-1}(x)}{n(N^{-1}(x))}\alpha(x)\,dx\Bigr) = \infty
\end{equation} has been used in \cref{theoParamT1}. A possible easier
condition could be \[\lim_{\delta\to1-}\alpha(\delta)\leq C<1\] but this
is sufficient and not necessary. Indeed, choosing
\(\alpha(\delta) = 1-\frac{c}{N^{-1}(\delta)}\) for a positive constant
\(c\leq N^{-1}(\tilde\delta)\), would still satisfy conditions of
\cref{theoParamT1} even though \(\alpha(\delta)\) has right limit equal
to \(1\). Firstly, observe that \begin{align*}
\int_{\tilde\delta}^{\delta}\frac{N^{-1}(x)}{n(N^{-1}(x))}\alpha(x)\,dx &= \int_{\tilde\delta}^{\delta}\frac{N^{-1}(x)}{n(N^{-1}(x))}\Bigl(1-\frac{c}{N^{-1}(\delta)}\Bigr)\,dx\\
&= \frac{N^{-1}(\delta)}2\bigl(N^{-1}(\delta)-2c\bigr) + d
\end{align*} where
\(d=-\frac{N^{-1}(\tilde\delta)}2\bigl(N^{-1}(\tilde\delta)-2c\bigr)\)
is a finite constant, and the integral diverges for \(\delta\) going to
\(1\). Also, the argument of the limit in \cref{eqAlphaLim} becomes
\[\frac{N^{-1}(\delta)^2}2 - \frac{N^{-1}(\delta)}2\bigl(N^{-1}(\delta)-2c\bigr) - d = cN^{-1}(\delta) - d\]
which diverges at \(1\).

	\hypertarget{practical-calibration-of-a-smile-in-delta}{%
\subsubsection{Practical calibration of a smile in
delta}\label{practical-calibration-of-a-smile-in-delta}}

	We now reconsider the calibration of a smile in delta started in
\cref{calibration-of-a-smile-in-delta-in-sigma_deltato-k}. The aim here
is to calibrate a delta smile which can be transformed into a smile in
strike and vice-versa and such that its transformation eventually
satisfies the two Fukasawa necessary conditions of no arbitrage defining
the set \(\Sigma_{\text{WA}}\) as in \cref{eqSigmaWA}.

	There are two methodologies that can be designed. The first one does not
guarantee that the calibrated smile lives in \(\Sigma_{\text{WA}}\), but
it guarantees that it lives in \(\Sigma_{\delta\to k}\). This means that
the smile in delta can be transformed into a smile in strike and the
latter satisfies the first Fukasawa necessary condition of no arbitrage
under hypothesis in \cref{remarkHypothesis} (i.e.~the function
\(d_1(k)\) is decreasing and surjective).

This methodology follows the steps:

\begin{enumerate}
\def\labelenumi{\arabic{enumi}.}
\tightlist
\item
  consider the market discrete pillars \(\{k_i,\sigma_i\}_i\);
\item
  convert them to the pillars \(\{\delta_i,\sigma_i\}_i\) by defining
  \(\delta_i= N(d_1(k_i,\sigma_i))\);
\item
  compute the pillars \(\{\delta_i,l_i\}\) with \(l_i=-k_i\);
\item
  given the pillars in point 3., interpolate/extrapolate a function
  \(\delta\to l(\delta)\) such that

  \begin{itemize}
  \tightlist
  \item
    \(l(0)=-\infty, l(1)=+\infty\),
  \item
    \(l\) strictly increasing,
  \item
    \(l\bigl(\frac12\bigr)<0\),
  \item
    \(l(\delta)\leq\frac{N^{-1}(\delta)^2}{2} \, \forall\delta>\frac12\),
  \item
    \(\exists! \tilde\delta |\, l(\tilde\delta)=\frac{N^{-1}(\tilde\delta)^2}{2}\).
  \end{itemize}
\end{enumerate}

This would guarantee that the smile \(\delta\to\sigma(\delta)\) defined
as \begin{equation*}
\sigma(\delta)\sqrt{T} =
\begin{cases}
N^{-1}(\delta) + \sqrt{N^{-1}(\delta)^2-2l(\delta)} & \text{if $\delta\leq\tilde\delta$},\\
N^{-1}(\delta) - \sqrt{N^{-1}(\delta)^2-2l(\delta)} & \text{if $\delta>\tilde\delta$}.\\
\end{cases}
\end{equation*} lives in \(\Sigma_{\delta\to k}\). In order to have that
the smile lives in \(\Sigma_\text{WA}\), i.e.~that the corresponding
smile in strike satisfies the two Fukasawa necessary conditions of no
arbitrage, we should add in step 4. the requirements:

\begin{itemize}
\tightlist
\item
  \(l'(\delta)>\frac{N^{-1}(\delta)}{n(N^{-1}(\delta))}\) a.s. for
  \(\delta\in\bigl]\frac12,\tilde\delta\bigr[\),
\item
  \(l'(\delta)<\frac{N^{-1}(\delta)}{n(N^{-1}(\delta))}\) a.s. for
  \(\delta>\tilde\delta\),
\end{itemize}

so that the function \(m(\delta)\) is increasing and surjective.

	Interpolating a function \(l\) which satisfies all the above
requirements is not immediate. For this reason, a second more cunning
calibration methodology can be implemented, using \cref{theoParamT1}.
The target of such calibration routine are the functions \(\lambda\),
\(\mu\) and \(\alpha\). These functions must satisfy the requirements in
\cref{theoParamT1} in order to guarantee that the smile \cref{eqSigmaT1}
is a.s. differentiable and can be transformed into a smile in
log-forward moneyness satisfying the conditions of bijectivity of the
functions \(d_1(k,\hat\sigma(k))\) and \(d_2(k,\hat\sigma(k))\).

The steps to be performed become:

\begin{enumerate}
\def\labelenumi{\arabic{enumi}.}
\tightlist
\item
  consider the market discrete pillars \(\{k_i,\sigma_i\}_i\);
\item
  convert them to the pillars \(\{\delta_i,\sigma_i\}_i\) by defining
  \(\delta_i= N(d_1(k_i,\sigma_i))\);
\item
  given the pillars in point 2., interpolate/extrapolate a function
  \(\delta\to \sigma(\delta)\) defined as in \cref{eqSigmaT1} such that

  \begin{itemize}
  \tightlist
  \item
    \(\tilde\delta\in\bigl]\frac12,1\bigr[\);
  \item
    \(\lambda\) is a positive function defined on
    \(\bigl]0,\frac12\bigr]\) such that
    \(\int_0^{\frac12}\lambda(x)\,dx = \infty\);
  \item
    \(\mu\) is a positive function defined on
    \(\bigl[\frac12,\tilde\delta\bigr[\);
  \item
    \(\alpha\) is a function defined on \(]\tilde\delta,1[\) such that
    \(\alpha(\delta)\in]0,1[\),
    \(\int_{\tilde\delta}^{1}\frac{N^{-1}(x)}{n(N^{-1}(x))}\alpha(x)\,dx = \infty\),
    and
    \(\lim_{\delta\to1^-}\bigl(\frac{N^{-1}(\delta)^2}2 - \int_{\tilde\delta}^{\delta}\frac{N^{-1}(x)}{n(N^{-1}(x))}\alpha(x)\,dx\bigr) = \infty\).
  \end{itemize}
\end{enumerate}

	\hypertarget{examples-of-smiles-in-sigma_textwa}{%
\subsubsection{\texorpdfstring{Examples of smiles in
\(\Sigma_\text{WA}\)}{Examples of smiles in \textbackslash Sigma\_\textbackslash text\{WA\}}}\label{examples-of-smiles-in-sigma_textwa}}

	\hypertarget{the-bounded-and-flat-smiles}{%
\paragraph{The bounded and flat
smiles}\label{the-bounded-and-flat-smiles}}

	Let us look at the requirement on \(\alpha\) detailed in
\cref{requirements-on-parameters}. It has been shown that
\(\alpha(\delta) = 1-\frac{c_\alpha}{N^{-1}(\delta)}\) with
\(c_\alpha\leq N^{-1}(\delta)\) satisfies conditions of
\cref{theoParamT1}. An interesting consequence of this example is that
the limit of \(\sigma(\delta)\) in \(1\) is finite. Indeed, for
\(\delta>\tilde\delta\), it holds
\[\sigma(\delta)\sqrt{T} = N^{-1}(\delta) - \sqrt{N^{-1}(\delta)\bigl(N^{-1}(\delta)-2c_\alpha\bigr)+2d}\]
and this in turn coincides with
\[\frac{2c_\alpha N^{-1}(\delta) -2d}{N^{-1}(\delta) + \sqrt{N^{-1}(\delta)\bigl(N^{-1}(\delta)-2c_\alpha\bigr)+2d}}\]
which converges to \(c_\alpha\) as \(\delta\) goes to \(1\). This means
that it is possible to obtain bounded smiles on the right appropriately
choosing the function \(\alpha(\delta)\) in the parametrization
\cref{eqSigmaT1}.

Similarly, it is possible to have bounded wings on the left choosing a
suitable \(\lambda(\delta)\) function. For example, we can define
\(\lambda(\delta) = \frac{c_\lambda}{n(N^{-1}(\delta))}\) to have the
convergence of the smile to \(c_\lambda\) on the left.

\Cref{figureBoundedSmile} shows a classical skew smile, which notably
has a bounded left wing of the smile in delta. The function \(\lambda\)
is defined as above with \(c_\lambda = 0.1\). In order to guarantee
continuity of the derivative and a nice shape of the smile, other
parameters have been chosen as \begin{align}\label{eqBoundedSmile}
&
\mu(\delta)=\frac{c_\lambda}{n(0)}\frac{\tilde\delta-\delta}{\tilde\delta-\frac12}
& &
\alpha(\delta) = \begin{cases}
\frac{n(N^{-1}(\delta))}{N^{-1}(\delta)}\mu(2\tilde\delta-\delta) & \text{if $\delta<\hat\delta$}\\
\frac{n(N^{-1}(\hat\delta))}{N^{-1}(\hat\delta)}\mu\bigl(\frac12\bigr)  & \text{if $\delta\geq\hat\delta$}
\end{cases}
\end{align} where \(\hat\delta=2\tilde\delta-\frac12\) and
\(\tilde\delta=0.7\).

\begin{figure}
	\centering
	\begin{subfigure}{.5\textwidth}
		\centering
		\includegraphics[width=.8\linewidth]{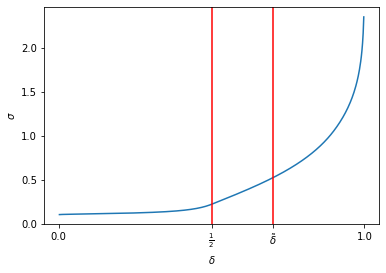}
	\end{subfigure}%
	\begin{subfigure}{.5\textwidth}
		\centering
		\includegraphics[width=.8\linewidth]{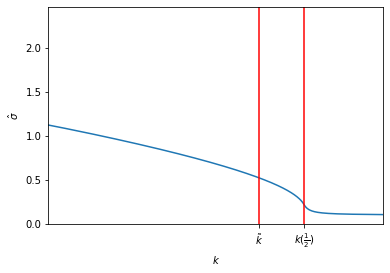}
	\end{subfigure}
	\caption{Skew shaped smile with bounded left wing in delta (left) and bounded right wing in log-forward moneyness (right) obtained with parameters as in \cref{eqBoundedSmile}}
	\label{figureBoundedSmile}
\end{figure}
{ \hspace*{\fill} \\}
    
	This example can be further pushed to obtain a flat smile. Indeed, if we
want a flat total implied volatility \(\sigma(\delta)\sqrt{T}\) at a
level \(c\), we can define \begin{align*}
&
\lambda(\delta) = \frac{c}{n(N^{-1}(\delta))}
& &
\mu(\delta)=\frac{c-N^{-1}(\delta)}{n(N^{-1}(\delta))}
& &
\alpha(\delta) = 1-\frac{c}{N^{-1}(\delta)}
\end{align*} and \(\tilde\delta = N(c)\).

	\begin{remark}

Smiles of the form \cref{eqSigmaT1} allow for bounded wings and for flat shapes.

\end{remark}

	\hypertarget{the-w-shaped-smile}{%
\paragraph{The W-shaped smile}\label{the-w-shaped-smile}}

	The parametrization \cref{eqSigmaT1} can be used to model very different
kind of smiles, and also odd ones. For example, we can model `sad
smiles' defining \(\tilde\delta=0.7\), \begin{align}\label{eqSadSmile}
&
\lambda(\delta) = 
\begin{cases}
\frac{\hat\delta^2c}{x^2} & \text{if $\delta<\hat\delta$}\\
c & \text{if $\delta\geq\hat\delta$}
\end{cases}
& &
\mu(\delta)=c-\frac{N^{-1}(\delta)}{n(N^{-1}(\delta))}
& &
\alpha(\delta) = 
\begin{cases}
1-c\frac{n(N^{-1}(\delta))}{N^{-1}(\delta)} & \text{if $\delta<\hat{\hat\delta}$}\\
1-c\frac{n(N^{-1}(\hat{\hat\delta}))}{N^{-1}(\hat{\hat\delta})}  & \text{if $\delta\geq\hat{\hat\delta}$}
\end{cases}
\end{align} where
\(c=\frac{N^{-1}(\tilde\delta)}{n(N^{-1}(\tilde\delta))}\),
\(\hat\delta=0.02\) and \(\hat{\hat\delta}=0.9\).

With these parameters, all conditions of \cref{theoParamT1} are
satisfied and the resulting smile in delta \(\sigma(\delta)\) has a
W-shape as in \Cref{figureSadSmile}.

\begin{figure}
	\centering
	\begin{subfigure}{.5\textwidth}
		\centering
		\includegraphics[width=.8\linewidth]{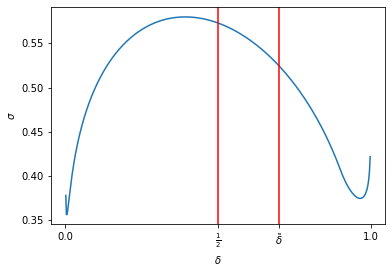}
	\end{subfigure}%
	\begin{subfigure}{.5\textwidth}
		\centering
		\includegraphics[width=.8\linewidth]{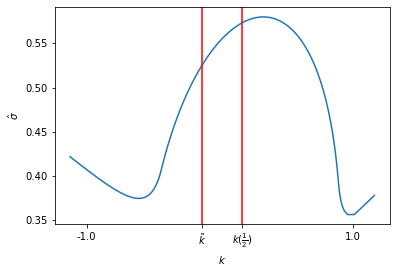}
	\end{subfigure}
	\caption{W-shaped smile in delta (left) and in log-forward moneyness (right) obtained with parameters as in \cref{eqSadSmile}}
	\label{figureSadSmile}
\end{figure}
{ \hspace*{\fill} \\}
    
	It is easy to show that the left and right limits of the smile (both in
delta and in strike) are infinite. Choosing different values for
\(\hat\delta\) and \(\hat{\hat\delta}\) allows to move the location of
the two minima and the maxima of the smile. In this way, it is possible
to obtain smiles with W-shapes that have been described in the
log-normal mixture framework by Glasserman and Pirjol
\cite{glasserman2021w} and have been seen, for example, for AMZN on the
\(26\) of April \(2018\) for options with expiry \(27\) April \(2018\),
before to the first quarter earnings announcement.

	\hypertarget{svi}{%
\paragraph{SVI}\label{svi}}

	The SVI model has been introduced by Gatheral at the Global Derivatives
conference in Madrid in 2004 \cite{gatheral2004parsimonious}. It is a
model for the implied total variance \(\hat\omega(k)=\hat\sigma(k)^2T\)
as a function of the log-forward moneyness \(k\) and it is defined as
\[\hat\omega(k) = a + b\bigl(\rho(k-m) + \sqrt{(k-m)^2+\bar\sigma^2}\bigr).\]

We suppose that the SVI parameters nder study satisfy the conditions of
a decreasing and surjective \(d_1(k)\) function and of a decreasing and
surjective \(d_2(k)\) function. In such way, the corresponding smile in
delta obtained through the definition
\(\sigma(\delta)=\hat\sigma(k(\delta))\) where
\(k(\delta)=N^{-1}(d_1^{-1}(\delta))\), belongs to \(\Sigma_\text{WA}\).
The SVI in the delta parameterization is not recovered in the present
article, however we can still do some tool computations on such smile.

For \cref{lemmaT1SigmaSol}, there exists a unique \(\tilde\delta\) such
that \(l(\tilde\delta)=\frac{N^{-1}(\tilde\delta)^2}{2}\). In the strike
notation, this is equivalent to say that there exists a unique
\(\tilde k\) such that \(-\tilde k = \frac{d_1(\tilde k)^2}{2}\), or
simplifying \(\sigma(\tilde k)^2T = -2\tilde k\). We now calculate such
\(\tilde k\).

	We need to look at the solutions of
\[a + b(\rho(k-m) + \sqrt{(k-m)^2+\bar\sigma^2}) = -2k,\] or
equivalently \[-(2+b\rho)k+b\rho m-a=b\sqrt{(k-m)^2+\bar\sigma^2}.\]
Under the Lee moment formula for \(\delta\gg\frac{1}{2}\) (or
\(k\ll0\)), it holds \(b(1-\rho)<2\), so \(2+b\rho>2+b\rho-b>0\). Then,
the above condition is never satisfied if
\(k\geq\frac{b\rho m-a}{2+b\rho}:=E\). Otherwise, we can take the square
and simplifying, one recovers a second-degree equation of the form
\(Ak^2+Bk+C=0\) where \begin{align*}
A &:= (2+b\rho-b)(2+b\rho+b)\\
B &:= 2\bigl(a(2+b\rho)-m\bigl(b^2(1-\rho^2)-2b\rho\bigr)\bigr)\\
C &:= (b\rho m-a)^2-b^2(m^2+\bar\sigma^2)=0.
\end{align*}

	The leading coefficient \(A\) is positive, and the Delta of such
equation is
\(\Delta = 4b^2((a+2m)^2+\bar\sigma^2(2+b\rho-b)(2+b\rho+b))\), which is
also positive since both terms are positive. Let us call \(k_+\) and
\(k_-\) the two possible solutions, with \(k_-<k_+\). They are
acceptable iff they are smaller than \(E\), or iff
\(\pm\sqrt{\Delta}<2AE+B\) respectively. The RHS is
\(\frac{2b^2(a+2m)}{2+b\rho}\), which is positive iff \(a>-2m\). In such
case, the \(+\) solution is acceptable iff
\[0<(2AE+B)^2-\Delta=2A(2AE^2-EB+2C)=-\frac{4b^2A}{(2+b\rho)^2}\bigl((a+2m)^2+\bar\sigma^2(2+b\rho)^2\bigr)\]
which is not possible. On the other hand, the \(-\) solution is
acceptable iff \(a>-2m\) or \(\Delta-(2AE+B)^2>0\), which is always
verified as proved above.

	In particular,
\[\tilde k = \frac{bm(2\rho-b(1-\rho^2))-a(2+b\rho) - b\sqrt{(a+2m)^2+\bar\sigma^2(2+b(1+\rho))(2-b(1-\rho))}}{(2+b(1+\rho))(2-b(1-\rho))}.\]

	The SVI model has given birth to other sub-models, obtained reducing the
original \(5\) parameters model to a model with less parameters. Among
them, the SSVI model by Gatheral and Jacquier
\cite{gatheral2014arbitrage} has been largely used in industry. It has
the form \begin{equation*}
\hat\omega(k) = \frac{\theta}{2}\bigl(1+\rho\varphi k + \sqrt{(\varphi k + \rho)^2 + (1-\rho^2)}\bigr).
\end{equation*} where the parameters are defined from the SVI ones as
\begin{align*}
\varphi = \frac{\sqrt{1-\rho^2}}{\sigma}, \quad \theta=\frac{2b\sigma}{\sqrt{1-\rho^2}}.
\end{align*}

In the case of SSVI, the expression for \(\tilde k\) is easier. Indeed
\[\tilde k = -\frac{2\theta}R\] where
\[R = \Bigl(2+\frac{\theta\varphi}2(1+\rho)\Bigr)\Bigl(2-\frac{\theta\varphi}2(1-\rho)\Bigr)\]
is positive for the Lee bounds, which require
\(\frac{\theta\varphi}2(1+|\rho|)<2\). The corresponding delta is
\[\tilde\delta = N\Bigl(4\sqrt{\frac{\theta}R}\Bigr).\]

	\hypertarget{conclusion}{%
\section{Conclusion}\label{conclusion}}

	The ability to pass from a smile in delta to a smile in strike and
vice-versa has been characterized requiring that the \(d_1\) function of
the Black-Scholes formula has to be decreasing and surjective. This
condition is one of the two necessary requirements for the absence of
butterfly arbitrage obtained by Fukasawa. Adding the second requirement
that the \(d_2\) function is decreasing ensures that also the symmetric
smile has the property of being transformed into the delta space.

The requirements that the \(d_1\) and \(d_2\) functions have to be
decreasing under no butterfly arbitrage can be translated into the delta
space with specific conditions. These conditions identify a
characterization of the set of smiles in delta satisfying the weak no
butterfly arbitrage requirements and allow to parametrize such set. As a
consequence, we have obtained a parameterization depending on one real
number and three positive functions which guarantees that the resulting
smiles in delta satisfy the weak no arbitrage conditions identified by
Fukasawa.

Practitioners who use smiles in delta could use those parameterizations
to ensure at least weak no butterfly arbitrage. An open challenging task
is to characterize the subfamily of no butterfly arbitrage smiles in
delta.

The task of characterizing the set of butterfly arbitrage-free smiles is
open in both delta and strike/log-forward moneyness spaces. The results
in the present article give hope of achieving the characterization of
such set using the delta parametrization.


\newpage \bibliography{Biblio}
\bibliographystyle{plain}

\end{document}